\documentclass[11pt]{amsart}
\usepackage{amsmath}
\usepackage{rotating}
\usepackage{longtable}
\usepackage{enumerate}
\usepackage{rotating}
\usepackage{color}

\usepackage{latexsym,array,delarray,amsthm,amssymb,epsfig,tabu,amsmath}

\usepackage{arydshln}
\usepackage{caption,color}
\usepackage{subcaption}
\usepackage{blkarray}
\usepackage{float}
\theoremstyle{plain}
\newtheorem{thm}{Theorem}[section]

\newtheorem{op}[thm]{Proposition}
\newtheorem{cor}[thm]{Corollary}

\newtheorem*{thm*}{Theorem}
\newtheorem*{lemma*}{Lemma}
\newtheorem*{prop*}{Proposition}
\newtheorem*{cor*}{Corollary}
\newtheorem*{conj*}{Conjecture}

\theoremstyle{definition}

\theoremstyle{remark}

\newcommand{\ind}{\mbox{$\perp \kern-5.5pt \perp$}}

\begin{document}

\author{Colby Long and Laura Kubatko}

\title[Effects of Gene Flow on the Coalescent]{The effect of gene flow on coalescent-based species-tree inference}

\begin{abstract}
Most current methods for inferring species-level phylogenies under the coalescent model assume
that no gene flow occurs following speciation. While some studies have examined the impact of gene flow on 
estimation accuracy for certain methods (e.g., \cite{eckertcarstens2008,chungane2011,leacheetal2014}), limited analytical work has been undertaken to directly assess
the potential effect of gene flow across a species phylogeny.  In this paper, we consider a three-taxon isolation-with-migration model that allows gene flow between sister taxa for a brief period following speciation, 
as well as variation in the effective population sizes across the tree. 
We derive the probabilities of each of the three gene tree topologies under this model, 
and show that for certain choices of the gene flow and effective population size parameters, 
anomalous gene trees (i.e., gene trees that are discordant with the species tree but that have higher probability than the gene tree concordant with the species tree) exist. We characterize the region of parameter space producing anomalous trees, and show that the probability of the gene tree that is concordant with the species tree can be arbitrarily small. 
We then show that the SVDQuartets method is theoretically valid under the model of gene flow between
sister taxa. 
We study its performance on simulated data and compare 
it to two other commonly-used methods for
species tree inference, ASTRAL and MP-EST. 
The simulations show that ASTRAL and MP-EST can be
statistically inconsistent when gene flow is present, while SVDQuartets performs well, 
though large sample sizes may be required for certain parameter choices.
\end{abstract}

\maketitle

\section*{Introduction}

Nearly all of the currently-used methods for coalescent-based estimation of species trees assume that speciation occurs at
a precise instant of time, with immediate cessation of gene flow, and that gene flow between distinct species elsewhere on the phylogeny
does not occur.  While several
studies have examined the impact of gene flow on the performance of species-level phylogenetic inference (e.g., \cite{eckertcarstens2008,chungane2011,leacheetal2014}), analytic 
analyses of the potential impact of gene flow are limited, despite the development of mathematical tools that could facilitate such study
(e.g., \cite{hobolthetal2011,andersenetal2014,tiankubatko2016}).  Recently, Tian and Kubatko (2016) \nocite{tiankubatko2016} considered a 
three-species IM (isolation-with-migration) model  \cite{heyneilsen2004,hey2010,wanghey2010} that allowed gene flow between sister species, and derived the probability distribution 
of gene tree topologies and gene tree histories under this model.  Their work generalized the work of Zhu and Yang (2012), who considered 
a similar model but with gene flow allowed only between the terminal sister taxa. Both Zhu and Yang (2012) \nocite{zhuyang2012}
and Tian and Kubatko (2016) used their results to estimate model parameters via maximum likelihood. 
Zhu and Yang (2012) considered sequence data arising from the model, while Tian and Kubatko (2016) considered the observed distribution of gene tree histories. 

Though Tian and Kubatko (2016) included some exploration of the effect of different choices of parameters on the gene tree topology distribution, they
focused on cases of symmetric migration
 in which the species tree satisfied the molecular clock and in which the effective population size (which determines the rate of coalescence) in each branch was the same.  
Under those conditions, they found that the gene tree topology
that matched the species tree occurred with probability at least as large as the other two gene tree topologies.  In other words, there were no
{\itshape anomalous gene trees} \cite{degnansalter2005} in this case.  The assumption of common effective population sizes is implicitly made in some species tree estimation 
methods \cite{liuetal2010}, while Bayesian approaches (e.g., \cite{heleddrummond2010}) 
generally include inference of separate effective
population size parameters in the 
posterior distribution they estimate.  
However, the effect of variation in the rate of coalescence on the distribution of gene tree topologies has not been carefully
examined,  in either the presence or the absence of gene flow (note that Degnan and Salter (2005) \nocite{degnansalter2005} also assumed the same
effective population sizes in all branches of the species tree). 

Here, we consider the distribution of gene tree topologies for a three-species IM model that generalizes that considered by Tian and Kubatko (2016) to include time 
periods with and without gene flow between both pairs of sister species. We further consider the effect of variation in population sizes and migration rates on 
this distribution.  We mathematically derive explicit conditions under which anomalous gene trees do exist, and show that, in fact, the probability of the
gene tree topology that matches the species tree can be made arbitrarily small.  We then consider the impact of these results on several 
of the commonly-used methods for species tree inference under the multispecies coalescent, including SVDQuartets \cite{chifmankubatko2014}, ASTRAL \cite{mirarabetal2014}, and MP-EST \cite{liuetal2010}.  We show that
SVDQuartets is theoretically valid under our general model, while ASTRAL and MP-EST appear to be statistically inconsistent in the presence of gene flow for some choices of parameters.  We compare the performance of all three methods using simulation for 
multilocus data, and additionally assess the performance of 
SVDQuartets for coalescent independent sites.  
We begin by carefully defining our model.

\section*{A General IM Model for Three Species}
Isolation-with-migration (IM) models have a long history in the population genetics literature \cite{heyneilsen2004,hey2010,wanghey2010}, and have more recently been considered in the context of species tree estimation
\cite{zhuyang2012,tiankubatko2016}.  Here we consider a more general version of the IM model than that presented in \cite{tiankubatko2016}.   Figure 
\ref{fig: IMmodel}  shows a species tree with three species, 
labeled $A$, $B$, and $C$,
with topology $((A,B),C)$. 
We assume that we have sampled one lineage from each species, and these lineages are
 labeled with lowercase letters $a$, $b$, and $c$.  
 Within each population
on the species tree, $\theta_{X}$ denotes the effective population size parameter, which determines the rate of coalescence, in the branch that represents species $X$
 ($X = A, B, C,$ or  $AB$, with $\theta_{AB}$ denoting the size of the population ancestral to species $A$ and $B$). 
 The parameter $\tau_1$ indicates the time from the present to the speciation event separating species 
 $A$ and $B$, and $\tau_2$ indicates the time from the present to the speciation event separating the ancestral species $AB$ from species $C$.  
 Following each speciation event (looking forward in time), there is a time
period over which migration occurs, indicated by blue shading in Figure \ref{fig: IMmodel}. 
The length of this time period of gene flow between species $A$ and $B$
is $t_{A \leftrightarrow B}$ and the length of the period of gene flow 
between the ancestor of $A$ and $B$ and species $C$ is $t_{AB \leftrightarrow C}$.
Thus, there are two intervals of time during which there is no gene flow (complete isolation), of lengths 
$\tau_1 - t_{A \leftrightarrow B}$ and $(\tau_2 - \tau_1) - t_{AB \leftrightarrow C}$.  
Finally, we define the rates of gene flow between species 
as $m_j$, where $j = 1$ for gene flow from $AB$ to $C$, $j=2$ for gene flow from $C$ to $AB$, $j=3$ for gene flow from $A$ to $B$, and $j=4$ for gene flow from $B$ to $A$ (looking
backward in time).

\begin{figure}[h]
    \centering
    \begin{subfigure}[b]{.5 \linewidth}
        \centering
	\includegraphics[height=4cm]{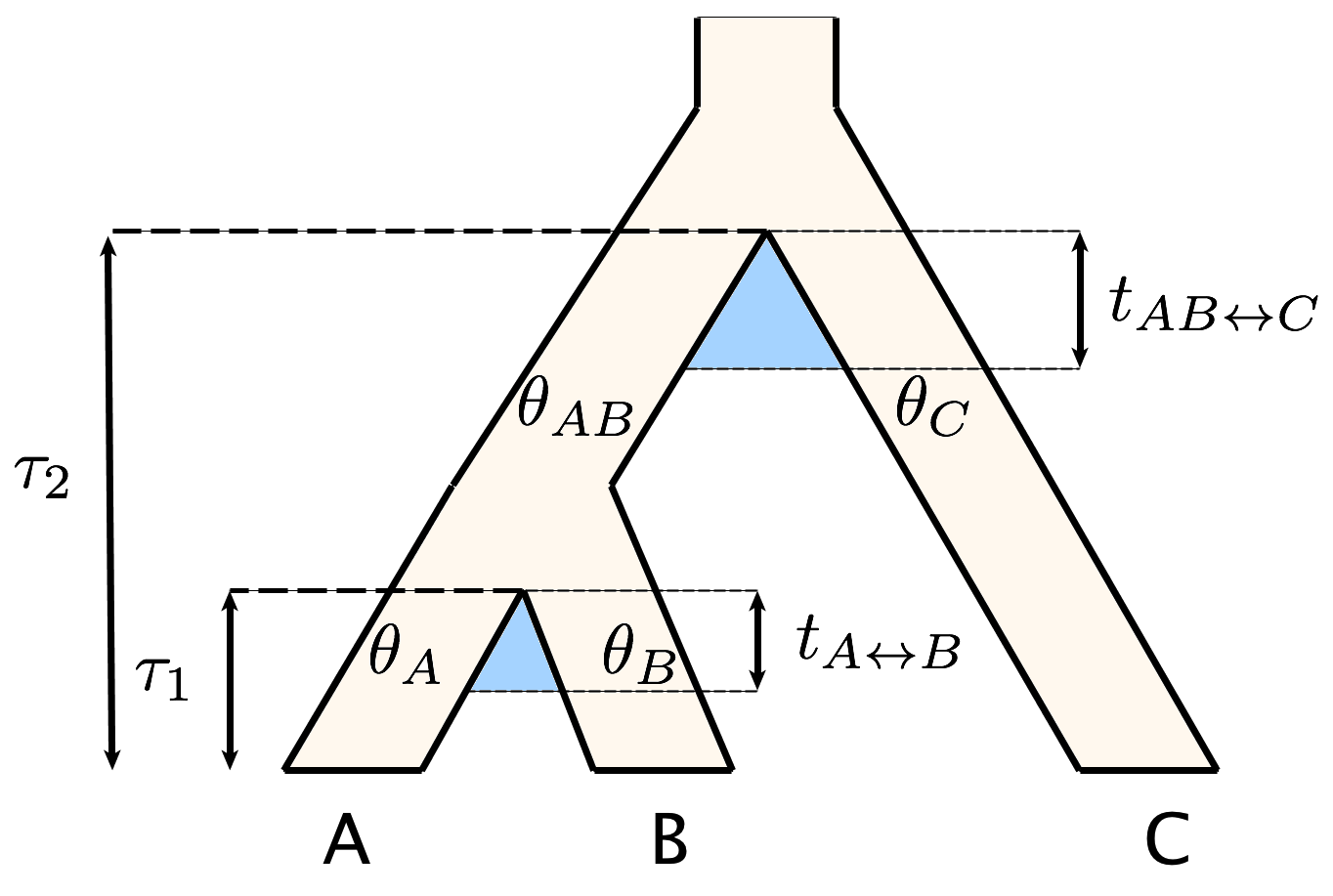}
         \caption{The 3-leaf IM model with gene flow.}
         \label{fig: IMmodel}
    \end{subfigure}%
    ~
    \begin{subfigure}[b]{.5 \linewidth}
        \centering
	\includegraphics[height=4cm]{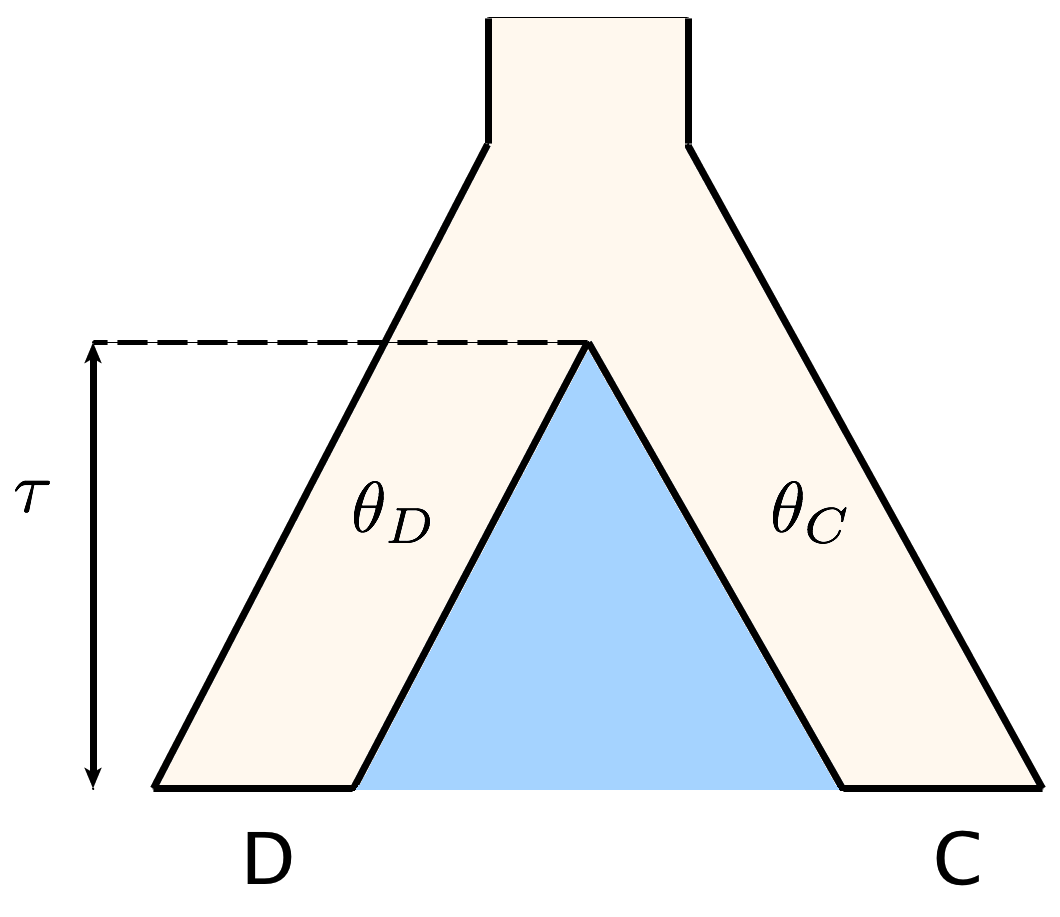}
	\caption{The reduced IM model with gene flow.}
	\label{fig: reducedIMmodel}
    \end{subfigure}
    \caption{The species tree parameters of the IM model with gene flow and
    the reduced IM model.}
\label{fig: SpeciesTrees}
\end{figure}

We are interested in computing the probabilities of the three possible gene tree topologies for different choices of the 
parameters defined above ($\theta_A, \theta_B, \theta_{AB}, \theta_{C}, m_1, m_2, m_3, m_4$, $t_{A \leftrightarrow B},
t_{AB \leftrightarrow C},\tau_1,$ and $\tau_2$).   We use lowercase letters $a$, $b$, and $c$ to denote lineages sampled from species $A$, $B$, and $C$, respectively, in the three 
possible gene trees.  The gene tree that matches the species tree is denoted by $((a,b),c)$, and will be referred to as the {\itshape concordant topology} (CT).  The two
{\itshape discordant topologies} (DT) are then $((b,c),a)$ and $((a,c),b)$.   
We begin by considering the effect of 
gene flow between the sister species $A$ and $B$. If there is any gene flow between these two species then there are only two possible outcomes.
First, the two lineages $a$ and $b$ could coalesce in either population $A$ or $B$, leading to a gene tree that matches the species tree.  Second, the lineages could fail to coalesce during the time over which $A$ and $B$ are separate species, resulting in both $a$ and $b$ entering the ancestral interval.   Varying the values of $\theta_A, \theta_B, m_3$ and $m_4$
will affect the probabilities of observing each of these outcomes, but regardless, it is impossible for coalescent events leading to discordant gene trees to occur before time $\tau_1$.
Likewise, for any fixed choice of $\theta_A, \theta_B, m_3$ and $m_4$, 
increasing $\tau_1$ will lead to a larger proportion of gene trees that match the species tree, 
while decreasing it will have the opposite effect. In either case, it is impossible to obtain
a gene tree discordant with the species tree as a result of any event that occurs before time $\tau_1$.

In order to understand the formation of anomalous gene trees, 
we will carefully analyze the second case in which $a$ and $b$ 
do not coalesce before $\tau_1$ and enter the ancestral interval. 
This allows us to study a reduced model with only two species, $C$ and $D$, where
$D$ represents the ancestor of $A$ and $B$ and contains two lineages, $a$ and $b$.  
In addition, since no discordant gene trees can form in the interval 
of no gene flow between population $AB$ and population $C$, we will exclude this interval.  
The reduced model over which we carry out calculations is depicted in Figure \ref{fig: reducedIMmodel}, where for simplicity we have labeled the time from the present to the speciation event $\tau$ and the population size 
for population $D$ by $\theta_D$. The migration parameters $m_1$ and $m_2$ are defined as before.   We note that the model in Figure \ref{fig: reducedIMmodel} can be thought of as a limiting case, obtained when $(\tau_2-\tau_1) - t_{AB \leftrightarrow C} \rightarrow 0$ and $m_3 = m_4 = 0$.   All of the results that we derive below can be straightforwardly extended to the model in Figure 1a.  We return to this point in the discussion at the end of the paper.

To compute gene tree topology probabilities under the model in Figure 1b, we use the approach introduced by Hobolth et al. (2011) in which we define a continuous-time 
Markov chain whose states are assignments of lineages to populations.  
For example, 
the ``starting state'' of our reduced Markov process (i.e., the state at the present time) is represented by $ddc$. What we mean by this is that the lineage sampled from species $A$ is in species $D$, the lineage sampled from species $B$ is in species $D$, and the lineage sampled from species $C$ is in species $C$.
In general, position in the string $xyz$ will denote whether we are referring to the lineage sampled from $A$, $B$, or $C$, respectively, and the letter assigned ($c$ or $d$) will indicate which species the lineage is a member of at the particular time under consideration.

A sequence of transitions between states corresponds to a {\itshape gene tree history}, a term used by both \cite{degnansalter2005} and \cite{tiankubatko2016} to denote the sequence and relative timing of a set of coalescent events that generate a gene tree. However, since we are only interested in the gene tree topology, not the specific
history that produced that topology, we will not need to keep track of the specific populations in which coalescent events 
occur. Instead, since the first coalescent event uniquely determines the topology of the gene tree that forms, we will introduce two absorbing states to the Markov chain that describes our model. These states, $CT$ and $DT$, correspond to the formation of concordant triples and discordant triples.
For example, suppose we begin in state $ddc$ and that the lineage sampled in species $B$ migrates to species $C$ at time $t_1$. 
Then we will have observed a transition of the Markov chain from state $ddc$ to state $dcc$. 
If the lineages in $C$ then coalesce at time $t_2$, then a discordant triple will eventually form. 
Therefore, the Markov chain transitions from state $dcc$ to the state $DT$. 
This sequence of transitions is depicted in Figure \ref{fig: MarkovChainExample}.

Once we have defined the Markov chain, we can compute the \emph{transition probabilities} for the reduced
IM model, the probability that the chain transitions from any state to any other state over the time interval of length $\tau$.  
Since we are particularly interested in the gene tree topology frequencies and since our chain always begins in state
$ddc$, we are most interested in the probabilities of transition from state $ddc$ to states $CT$ and $DT$. 
Notice, however, that some gene tree must form even if no coalescent event occurs in the gene flow
interval. If no coalescent event occurs over the gene flow interval, then the Markov chain ends in one of the 
states $ccc,ccd,cdc,dcc,ddc,dcd,$ or $ddd$, and all three lineages enter the population ancestral to 
$C$ and $D$. In such a case, the probability of each gene tree topology forming is equal to 1/3. 
Therefore, knowing all of the transition probabilities will allow us to determine the probability of a concordant
triple or discordant triple forming in the reduced IM model. By symmetry in the model, the two discordant
topologies form with equal probability, and so we can compute each gene tree topology probability explicitly.
However, there is one further observation that we may make to reduce the complexity of the Markov 
chain representing the reduced IM model. That is that once the chain enters either state $ccc$ or $ddd$,
the probability of each gene tree topology forming is now 1/3, and we do not need to keep track of any other
subsequent transitions. Thus, for simplicity, we can treat $ccc$ and $ddd$ as absorbing states.

\begin{figure*}[h]
    \centering
    \begin{subfigure}[b]{.5 \linewidth}
        \centering
	\includegraphics[height=5.5
cm]{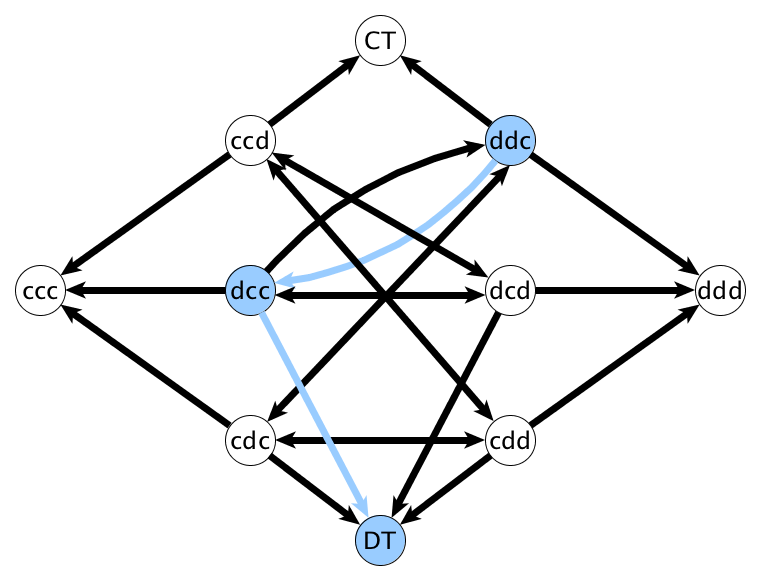}
         \caption{The Markov chain of the reduced IM model.}
         \label{fig: GeneFlowMarkovChain}
    \end{subfigure}%
    ~
    \begin{subfigure}[b]{.5 \linewidth}
        \centering
	\includegraphics[width=7cm]{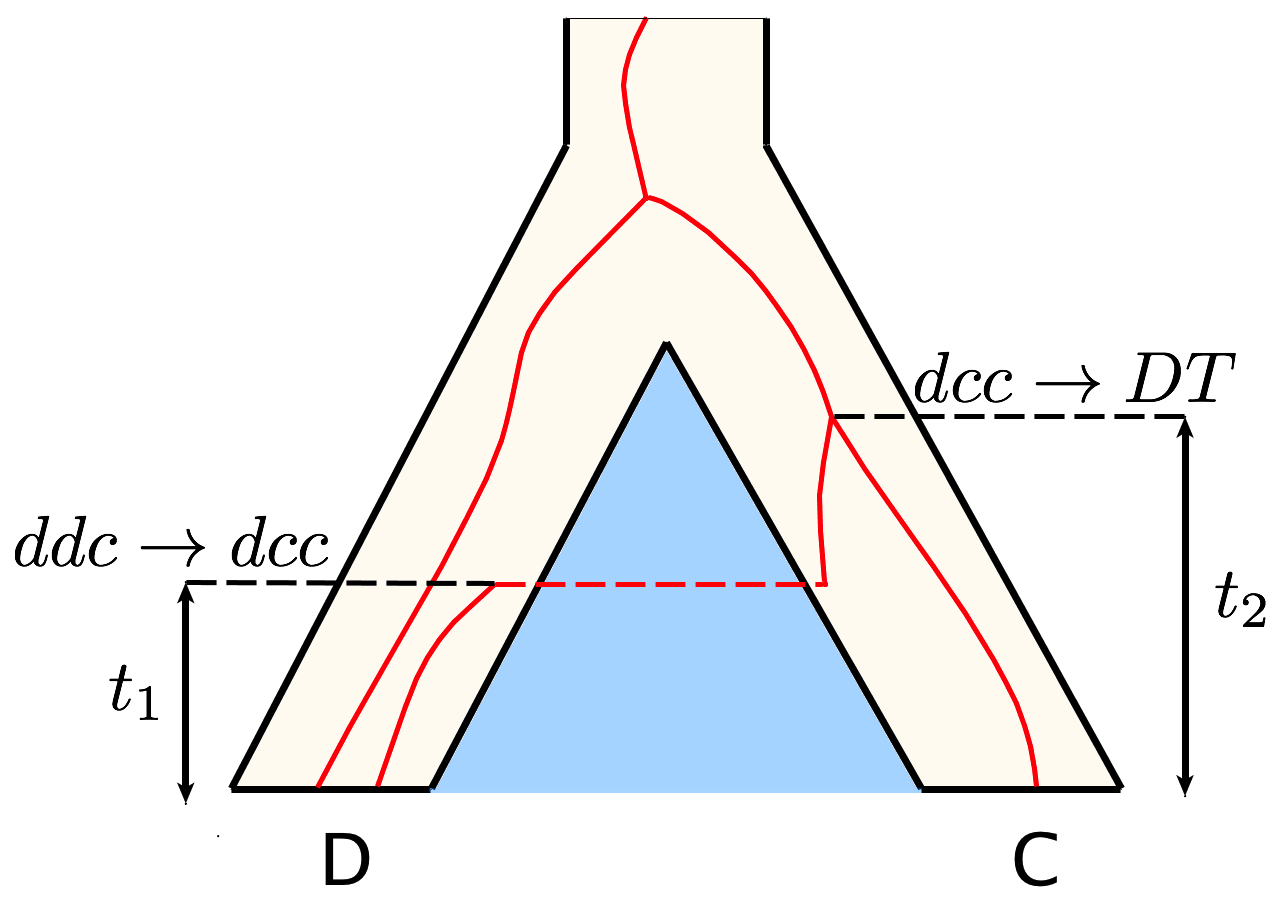}
	\vspace{0.5cm}
	\caption{The species tree of the reduced IM model.}
	\label{fig: CoalescentHistory }
    \end{subfigure}
    \caption{\small A particular sequence of coalescent and migration events in the reduced IM model
    and the corresponding path in the Markov chain. The chain starts in state $ddc$ (shaded circle in (a) and depicted 
    at the present time (i.e., at the tips of the tree) in (b)). At time $t_1$, the lineage sampled from $B$ migrates  between $D$ and $C$, 
    as depicted in (b) with a dotted horizontal line. This is represented in (a) by a transition from state $ddc$ to state $dcc$ (blue arrow). At time $t_2$,
    the lineages sampled from $B$ and $C$ coalesce in $C$,  which is a transition from $dcc$ to $DT$ (blue arrow in (b)).}
\label{fig: MarkovChainExample}
\end{figure*}

Figure \ref{fig: GeneFlowMarkovChain} depicts the Markov chain of the reduced IM model with arrows representing the possible transitions between states under the assumption that only a single event can happen at a given instant of time (i.e., two or more events do not 
happen at precisely the same instant of time).  
The instantaneous transition rates are the parameters of the model we described above, and the 
instantaneous rate matrix is given by,

$$
Q = 
\begin{blockarray}{ccccccccccc}
 & ddc & dcd & cdd  & ccd & cdc & dcc & ddd & ccc & CT & DT \\
\begin{block}{c(ccc|ccc|cccc)}
ddc  & -- & 0 & 0  & 0 & m_1 & m_1 & m_2  & 0 & \alpha_D   & 0 \\
dcd  & 0 & -- & 0  & m_1 & 0 & m_1 & m_2 & 0   & 0 &   \alpha_D   \\
cdd  & 0 & 0 & --   & m_1 & m_1 & 0 & m_2  & 0   & 0 &   \alpha_D \\
\cline{2-11}
ccd  & 0 & m_2 & m_2  & -- & 0 & 0 & 0  & m_1 &   \alpha_C   & 0 \\
cdc  & m_2 & 0 & m_2  & 0 & -- & 0 & 0  & m_1& 0   &   \alpha_C \\
dcc  & m_2 & m_2 & 0  & 0 & 0 & -- & 0  & m_1 & 0   &   \alpha_C \\
\cline{2-11}
ddd & 0 & 0 & 0 & 0 & 0 & 0 & 0 & 0 & 0& 0   \\
ccc & 0 & 0 & 0 & 0 & 0 & 0 & 0 & 0 & 0   & 0 \\
CT & 0 & 0 & 0 & 0 & 0 & 0 & 0 & 0 & 0   & 0 \\
DT & 0 & 0 & 0 & 0 & 0 & 0 & 0 & 0 & 0   & 0 \\
\end{block}
\end{blockarray}.
$$
where $\alpha_C = 2/\theta_{C}$ and $\alpha_{D} = 2/\theta_D$ are the rates of coalescence in populations $C$ and $D$, respectively.  
In the next section, we explain how to use the instantaneous rate matrix $Q$ to compute
the transition probabilities over the interval $\tau$ to derive explicit formulas for the 
probabilities of the three gene tree topologies.

 \section*{Analytic Results for the Isolation with Migration Model}

In this section, we derive an explicit formula for the probability of observing a gene 
tree concordant with the species tree in the reduced IM model described above.
As noted previously, since the other two gene tree topologies occur with equal probability,
this allows us to determine the probability of each gene tree topology. 
To obtain this formula, we must exponentiate the product of the rate matrix $Q$ and the parameter $\tau$ 
in terms of the parameters to obtain the transition probabilities \cite{karlintaylor1975}.
In general, computing the exponential of a matrix
in terms of the parameters is difficult, even for relatively small matrices.
However, the number of absorbing states in our simplified model allows us 
to write the rate matrix $Q$ as a block matrix with two zero blocks. 
This will enable us to give explicit formulas for entries of the matrix exponential, and 
consequently, to obtain an explicit formula for the probability of observing the concordant triple.

To compute the probability of observing a concordant triple in the 
reduced IM model, we will consider the two possible ways that one might form. 
One way is for the
Markov chain to transition from 
state $ddc$ to state $CT$
during the interval of gene flow of length $\tau$ in the species tree.
The other way for a concordant triple to form is for 
the Markov chain to transition from $ddc$ to 
any state other than $CT$ or $DT$ along the interval of gene flow.
As noted above, conditional on this event, the concordant triple will form with probability 1/3, since all three lineages
entered the same branch during the interval of gene flow or will enter the same branch
$CD$ at the end of the interval of gene flow. Thus, the concordant triple frequency
as a function of the rate matrix is given by,

\begin{align}
F(\alpha_C,\alpha_D,m_1,m_2,\tau) = exp(Q\tau)_{ddc,CT} + \frac{1}{3}(1 - exp(Q\tau)_{ddc,CT} - exp(Q\tau)_{ddc,DT}).
\end{align}

To expand the formula above in terms of the parameters, 
we only require two entries of the transition matrix. 
Rewriting the rate matrix $Q$ in block form 
will enable us to write the matrix exponential in terms of submatrices of $Q$. 
For the following proposition, let
$
Q = \begin{pmatrix}
M & N \\
0 &  0 \\
\end{pmatrix}
$
where \\
$$M = 
\begin{pmatrix}
 -- & 0 & 0  & 0 & m_2 & m_2  \\
 0 & -- & 0  & m_2 & 0 & m_2   \\
 0 & 0 & --   & m_2 & m_2 & 0  \\
 0 & m_1 & m_1  & -- & 0 & 0 \\
m_1 & 0 & m_1  & 0 & -- & 0 \\
m _1& m_1 & 0  & 0 & 0 & --  \\
\end{pmatrix}
\text{ and 
where } N = 
\begin{pmatrix}
 m_2  & 0 & \alpha_D   & 0 \\
 m_2  & 0   & 0 & \alpha_D   \\
 m_2 & 0   & 0 & \alpha_D \\
0  & m_2 & \alpha_C   & 0 \\
0  & m_2 & 0   & \alpha_C \\
0  & m_2 & 0   & \alpha_C \\
\end{pmatrix}.$$

 \begin{op} For any $\tau \in \mathbb{R}$ and generic choices of rate matrix parameters $\alpha_C,\alpha_D,m_1,$ and
 $m_2$,
$$exp(Q\tau) =
 \begin{pmatrix}
exp(M\tau) & ((exp(M\tau) - I)M^{-1}N) \\
0 &  0 \\
\end{pmatrix}.$$
\end{op}

\begin{proof}
Recall that the formula for the matrix exponential is given by 
$$exp(Q\tau) = \displaystyle \sum_{i=0}^\infty \frac{(Q\tau)^i}{i!} = I + Q\tau + \frac{(Q\tau)^2}{2!} + \ldots,$$
and by induction, for $i \geq 1$,
$(Q\tau)^i = 
\begin{pmatrix}
(M\tau)^i & M^{i-1}\tau^iN \\
0 &  0 \\
\end{pmatrix}. 
$
Therefore, the upper left block of $exp(Q\tau)$ is simply equal to $exp(M\tau)$.
We also have 
\begin{align*}
exp(M\tau) &= I + M\tau + \frac{M^2\tau^2}{2!} + \frac{M^3\tau^3}{3!} + \ldots \\
(exp(M\tau) - I) &=  M\tau + \frac{M^2\tau^2}{2!} + \frac{M^3\tau^3}{3!} + \ldots \\
(exp(M\tau) - I)M^{-1}N &= \tau N + \frac{M\tau^2N}{2!} + \frac{M^2\tau^3N}{3!} + \ldots, \\
\end{align*}
and this last line is exactly the upper right block of $exp(Q\tau)$.
\end{proof}

The formula for the probability of observing the concordant triple
depends only on two entries from the upper right block of the transition matrix $exp(Q\tau)$.
To compute these entries, we must first compute the $6 \times 6$ matrix $exp(M\tau)$. 
For a generic choice of nonnegative parameters, the matrix $M$ has only four distinct eigenvalues
but can be diagonalized so that $M = P\Lambda P^{-1}$,
where $\Lambda$ is the diagonal matrix of the eigenvalues of $M$.
Thus, $exp(M\tau) = Pexp(\Lambda \tau)P^{-1}$. Similarly, for 
a generic choice of nonnegative parameters we can write the entries of 
$M^{-1}$ explicitly, and simplify to write 
the concordant triple frequency as

 \begin{align*}
 F(\alpha_C,\alpha_D, & m_1,m_2,\tau) = -(1/3) (-2fm_1^2-2fm_2^2+fe_1\alpha_Dm_1+2fe_1\alpha_Dm_2+fe_2\alpha_C\alpha_D-fe_2\alpha_Cm_1+ \\
& fe_2\alpha_Dm_1+2fe_2\alpha_Dm_2-e_1\alpha_C\alpha_Dm_1+3e_1\alpha_C\alpha_Dm_2-3e_1\alpha_Cm_1m_2+e_1\alpha_Dm_1m_2+ \\
&e_2\alpha_C\alpha_Dm_1-3e_2\alpha_C\alpha_Dm_2+3e_2\alpha_Cm_1m_2-e_2\alpha_Dm_1m_2+fe_1\alpha_C\alpha_D-fe_1\alpha_Cm_1- \\
&4fm_1m_2+e_1\alpha_C^2\alpha_D-e_1\alpha_C^2m_1-e_1\alpha_C\alpha_D^2-3e_1\alpha_Cm_1^2-e_1\alpha_D^2m_1- \\
&2e_1\alpha_D^2m_2-e_1\alpha_Dm_1^2+2e_1\alpha_Dm_2^2-e_2\alpha_C^2\alpha_D+e_2\alpha_C^2m_1+e_2\alpha_C\alpha_D^2+ \\ &3e_2\alpha_Cm_1^2+e_2\alpha_D^2m_1+2e_2\alpha_D^2m_2+e_2\alpha_Dm_1^2-2e_2\alpha_Dm_2^2-3f\alpha_C\alpha_D-\\
&f\alpha_Cm_2-3f\alpha_Dm_1-6f\alpha_Dm_2)/ \\
&((\alpha_C\alpha_D+2\alpha_Cm_1+\alpha_Cm_2+\alpha_Dm_1+2\alpha_Dm_2+2m_1^2+4m_1m_2+2m_2^2)f)
\end{align*}

 where 
\begin{align*}
g &= -\frac{1}{2}(\alpha_C+\alpha_D+3m_1+3m_2) \\
f &= (\alpha_C^2-2\alpha_C\alpha_D-2\alpha_Cm_1+2\alpha_Cm_2+\alpha_D^2+2\alpha_Dm_1-2\alpha_Dm_2+m_1^2+2m_1m_2+m_2^2)^{1/2} \\
e_1 &= exp((g+f/2)\tau) \\
e_2 &= exp((g-f/2)\tau). \\
\end{align*}
 
All of the parameters of the model are nonnegative real numbers and the 
formula is valid and well-defined for a generic choice of parameters. 
The Maple file to verify this formula is contained in the Supplementary Materials.
There, we also give an example to show that this formula agrees with the output from 
the COALGF software \cite{tiankubatko2016} for the concordant triple frequency.

\subsection*{Characterizing the Gene Flow Anomaly Zone}

With this formula in hand, we can state a number of propositions about the 
concordant triple frequency and identify various regions of parameter space that 
produce anomalous behavior. Here, we define the \emph{gene flow anomaly zone} to be the region of parameter
space for which the concordant triple frequency is less than $1/3$. 
We note again that by the symmetry of the lineages $a$ and $b$ in the reduced IM model,
the two discordant triple frequencies are always given by $\frac{1}{2}(1 -  F(\alpha_C,\alpha_D,m_1,m_2,\tau))$.
The following proposition shows that this model is flexible enough to produce any concordant triple
frequency.

\begin{op} 
For any $\gamma \in (0,1)$  there exists a choice of parameters $(\alpha_C,\alpha_D,m_1,m_2,\tau) \in \mathbb{R}^5_{\geq 0}$ such that
$F(\alpha_C,\alpha_D,m_1,m_2,\tau) = \gamma$.
\end{op}

\begin{proof}
For fixed $m_1 >0$, setting $\alpha_C=\tau$ and $\alpha_D = m_2 = 1/\tau$, $\displaystyle \lim_{\tau\to \infty} F(\alpha_C,\alpha_D,m_1,m_2,\tau) = 0$. 
Similarly, for fixed $\alpha_C,\alpha_D>0$, if $m_1 = m_2 = 1/\tau$, then $\displaystyle \lim_{\tau\to \infty} F(\alpha_C,\alpha_D,m_1,m_2,\tau) = 1$.
\end{proof}

Intuitively, discordant triples are very likely to form when we allow unidirectional
gene flow from population $D$ to population $C$, reduce the probability of coalescence in 
population $D$, and then impose a very high rate of 
coalescence in $C$. Since lineages $a$ and $b$ are unlikely to coalesce in $D$,
given enough time, eventually one of these lineages will 
flow into population $C$ and immediately coalesce with lineage $c$ to produce
a discordant triple. While this scenario gives insight into what conditions might produce
anomalous behavior, it is not actually a requirement that the 
rate of gene flow from $D$ to $C$ be greater than the rate from $C$ to $D$ (i.e., that $m_1 > m_2$).
In \cite{tiankubatko2016}, the authors analyzed a model with equal rates of gene flow,
(i.e., where $m_1 = m_2$) and showed the concordant triple frequency could be 
made very close to $1/3$. In fact, the following proposition shows that 
even with this restriction, the concordant triple frequency can be made to approach $1/9$.

\begin{op} 
For any $\gamma \in (\frac{1}{9},1)$  there exists a choice of parameters $(\alpha_C,\alpha_D,m,m,\tau) \in \mathbb{R}^5_{\geq 0}$ 
such that
$F(\alpha_C,\alpha_D,m,m,\tau) = \gamma$.
Moreover, for any choice of parameters $(\alpha_C,\alpha_D,m,m,\tau) \in \mathbb{R}^5_{\geq 0}$, 
${1}/{9} \leq F(\alpha_C,\alpha_D,m,m,\tau) \leq 1$. 
\end{op}

\begin{proof}
From the initial state $ddc$, the sum of the instantaneous rates of change to the states $ddd$ and 
$CT$ is greater than ${1}/{3}$ the total rate of change to all other states. Once the system is in 
either $ddd$ or $CT$, the probability of a concordant triple forming is at least $1/3$. Therefore,
$F(\alpha_C,\alpha_D,m,m,\tau) \geq \frac{1}{9}$.

For fixed $m_1 = m_2> 0$, setting $\alpha_C=\tau$ and $\alpha_D = 1/\tau$, $\displaystyle \lim_{\tau\to \infty} F(\alpha_C,\alpha_D,m_1,m_2,\tau) = 1/9$. 
Similarly, for fixed $\alpha_C$ and $\alpha_D>0$, if $m_1 = m_2 = 1/\tau$ then $\displaystyle \lim_{\tau\to \infty} F(\alpha_C,\alpha_D,m_1,m_2,\tau) = 1$.
\end{proof}

\begin{figure}
\begin{center}
\includegraphics[width=9cm]{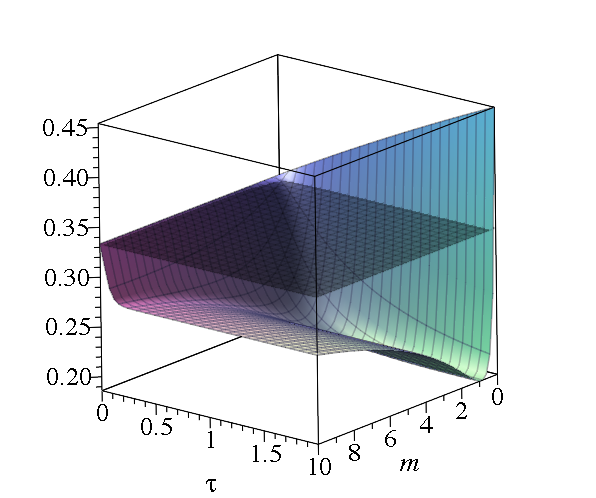}
\caption{A plot of the concordant triple frequency from the reduced IM model
as a function of the length of the
interval of gene flow ($\tau$) and the symmetric rate of gene flow ($m = m_1 = m_2$).
The other parameters of the model are fixed at $\alpha_C = 10$ and $\alpha_D = 0.1$.
Note that everywhere the surface is below the shadowed plane 
is a region of parameter space producing anomalous gene trees. }
\label{fig: SurfacePlot}
\end{center}
\end{figure}

Figure \ref{fig: SurfacePlot} shows a region of parameter space with $m_1 = m_2$
which produces anomalous gene trees and where the concordant triple frequency
approaches the limiting value of $1/9$. Whether or not the gene flow rates are equal, 
one essential feature for producing anomalous behavior is that
the rate of coalescence in population $C$ must be greater than that in population $D$. The following proposition makes 
this relationship between the coalescence parameters formal.

\begin{op} 
\label{op: c greater}
For any fixed $\alpha_C',\alpha_D'>0$, there exists an $m_1', m_2',\tau' \in \mathbb{R}_{\geq 0}$ such that 
$F(\alpha_C',\alpha_D',m_1',m_2',t') < 1/3$ if and only if $\alpha_C' > \alpha_D'$. 
\end{op}

\begin{proof} For all $\alpha_C,\alpha_D,m_1,m_2 >0$, $F(\alpha_C,\alpha_D,m_1,m_2,0) = 1/3$ and $\frac{\partial F}{\partial \tau}(0) = 2\alpha_D/3 > 0$. There is at most one point for which $\frac{\partial F}{\partial \tau}(s) = 0$, which implies that 
if there exists $m_1',m_2',t'$ such that $F(\alpha_C',\alpha_D',m_1',m_2',\tau') < 1/3$, it must be that 
$\displaystyle \lim_{\tau\to \infty} F(\alpha_C',\alpha_D',m_1',m_2',\tau) < 1/3$.
Evaluating the limit and solving, we see that $\displaystyle \lim_{\tau\to \infty} F(\alpha_C,\alpha_D,m_1,m_2,\tau)<1/3$ if and only if
$2\alpha_C\alpha_D + 2\alpha_Dm_1 + 4\alpha_Dm_2 - 2\alpha_Cm_1 < 0$. 
\end{proof}

\begin{cor}
For any fixed $\alpha_C',\alpha_D'>0$, there exists an $m', \tau' \in \mathbb{R}_{\geq 0}$ such that 
$F(\alpha_C',\alpha_D',m',m',\tau') < 1/3$ if and only if $\alpha_C' > 3\alpha_D'$. 
\end{cor}

\section*{Species Tree Estimation in the Presence of Gene Flow}

There are several phylogenetic estimation methods that reconstruct species
tree topologies under the coalescent model using the observed frequencies of estimated gene tree topologies.
These so-called \emph{summary methods} (or sometimes, {\itshape summary statistics methods} \cite{liuetal2009}) take advantage of the fact that for 
the coalescent model on an unrooted quartet tree, the most probable gene tree topology 
is that concordant with the species tree. By choosing an outgroup, these methods can be applied
to infer the rooted topology of a 3-leaf species tree under the coalecsent model. 
However, as we have just shown, in the presence of gene flow, the most common gene tree
topology may not be that concordant with the species tree topology. As was also noted
in \cite{solislemusetal2016}, the summary methods will be positively misleading in such cases. 
By contrast, in this section we show that the method of SVDQuartets remains theoretically valid
in the presence of gene flow. We offer here a brief review of the method which suggests why it 
should be effective on data produced by the IM model with gene flow described in this paper. 

\subsection*{Theoretical Justification for the Performance of SVDQuartets}\label{subsec:theory.svdq}

The method
of SVDQuartets is applied to sequence data generated under the multispecies coalescent in order
to infer the unrooted species tree topology \cite{chifmankubatko2014}. The method works by using singular value decomposition
to infer the unrooted topology of the induced 4-leaf subtrees of the species tree, 
also called \emph{quartets}, which are then assembled to reconstruct
the unrooted species tree. 
The method of inferring quartets is statistically consistent and can be combined with 
any exact algorithm for reconstructing species trees from quartets to yield
a statistically consistent method of inferring species trees under the coalescent 
model.

In this section we offer some justification for why SVDQuartets
is still an effective method for inferring the species tree under the three-species 
IM model with gene flow.
By selecting an outgroup, we can use SVDQuartets to infer the rooted
topology of a 3-leaf species tree under the coalescent model.
With the outgroup added, the 4-leaf species tree of the model is an equidistant, rooted caterpillar,
and we label the cherry of this tree by $A$ and $B$.
The symmetry in the cherry of the species tree implies that $p_{i_1i_2i_3i_4} = p_{i_2i_1i_3i_4}$, 
where $p_{i_1i_2i_3i_4}$ is the probability of observing the DNA site pattern $i_1i_2i_3i_4$, $i_j \in \{A, C, G, T\}$, $j = 1, 2, 3, 4$, at the tips
of the species tree under the coalescent model.
For the purposes of this paper, it will suffice to know that 
producing a probability distribution that satisfies these linear relationships is a sufficient but not necessary
condition for SVDQuartets to correctly infer the unrooted topology of the 4-leaf species tree with unlimited data.
We omit the full details of the SVDQuartets method here and refer the reader instead to
\cite{chifmankubatko2015}. 
Because of the
outgroup, the rooted species tree can be assumed to be a rooted caterpillar, and so 
the unrooted topology uniquely determines the rooted 3-leaf species tree.
Our claim is that the expected site pattern probabilities from the three-species
IM model with gene flow also satisfy these same linear relationships, 
implying that SVDQuartets is still valid in the presence of gene flow.

That the three-species IM model with gene flow satisfies these linear relationships 
follows simply by the symmetry between $A$ and $B$ in the model in the case where the 
rates of gene flow and coalescence between $A$ and $B$ are the same (i.e., $m_3 = m_4$ and $\theta_A =\theta_B$). 
 If $m_3 \not = m_4$ or $\theta_A \not = \theta_B$, 
 then the argument is only slightly more involved. 
We divide the gene trees into two sets, those that involve a coalescent event
in either $A$ or $B$ and those that do not.
The first set of gene trees are all concordant with the species tree. By the symmetry
in the cherry of each concordant gene tree,
the site pattern probabilities produced by each of these gene trees satisfies
$p_{i_1i_2i_3i_4} = p_{i_2i_1i_3i_4}$, and so must the sum of the 
site pattern probabilities from all of these gene trees.
For the second set of gene trees, each of the two lineages
$a$ and $b$ both enter the population $AB$. While the site pattern probabilities from
 each individual gene tree will not necessarily satisfy these relationships, the indistinguishability
 of the lineages $a$ and $b$ upon entering population $AB$ 
guarantees that the sum of the site pattern probabilities from all of these gene trees will
satisfy $p_{i_1i_2i_3i_4} = p_{i_2i_1i_3i_4}$, regardless of the values of $m_1$ and $m_2$.
Since the site pattern probabilities at the tips of the species tree under the three-species IM model are the sum of the 
site pattern probabilities from these two sets of gene trees, they must also satisfy $p_{i_1i_2i_3i_4} = p_{i_2i_1i_3i_4}$.  To formally prove the validity of SVDQuartets 
under this model requires a more detailed justification of these relationships and verification that the discordant topologies do not satisfy these relationships. Further, the result
can be extended to the case in which the species tree does not satisfy the molecular clock, following the 
argument in \cite{longkubatko2017}.

\subsection*{Performance of Reconstruction Methods in the Presence of Gene Flow}
Given the result above that indicates that SVDQuartets is theoretically valid in the presence of gene flow, we assess the performance of the 
method under various conditions using simulation.  For all of the simulations described here, the four-taxon model tree is shown in Figure \ref{fig:sim.model},
with speciation times given by $\tau_1 = 1.0$, $\tau_2 = 3.0$, and $\tau_3 = 10.0$, in coalescent units. The choice of a large 
interval between the second and third speciation events was made so that the gene tree topology 
distribution would mimic that of the three-taxon case described above.  Because the interval $\tau_3-\tau_2$
is long, we expect lineages $a$, $b$, and $c$ to coalesce completely within this interval with high probability. 
Thus the tree can be viewed as being rooted by taxon $R$.  We use this model tree to 
carry out two distinct sets of simulation studies that differ in the 
types of input data used. We describe each simulation set-up and the corresponding results separately below.

\begin{figure}
\begin{center}
\includegraphics[width=3.0in]{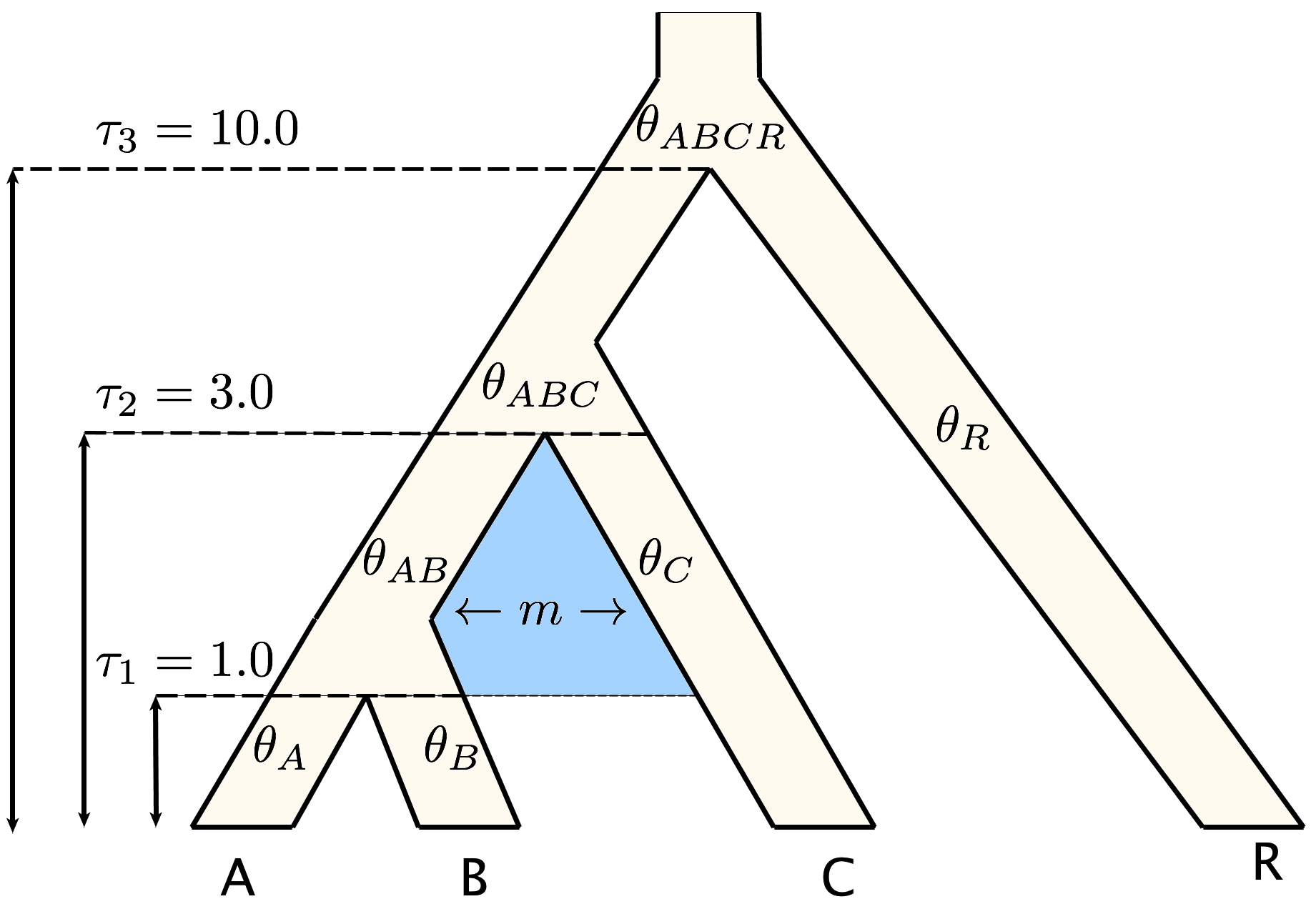}
\caption{Model tree for the simulation studies.}\label{fig:sim.model}
\end{center}
\end{figure}

\subsubsection*{Coalescent Independent Sites.}
Our first simulation study involves the generation of data consisting of {\itshape coalescent 
independent sites}.  Coalescent independent sites are sampled by first simulating some number, $n$, 
of gene trees under the model in Figure \ref{fig:sim.model}. For each of the $n$ gene trees,
a single site is generated according to one of the standard nucleotide substitution models. This results in a dataset consisting of 
$n$ sites that are conditionally independent given the species tree.  This simulation setting is designed to mimic SNP data, except that
there is no requirement that sites are variable or that they include only two alleles.  In other words, any site pattern can be included in the
dataset, but all sites are independently generated under the model, conditional on the species tree.
Within the setting of coalescent independent sites, we consider three distinct sets of simulations designed to assess the performance of SVDQuartets
under different scenarios.  The first setting involves conditions under which anomalous gene trees exist due to the presence of gene flow and
 varying $\theta$ parameters.  
 In particular, we fix $\theta_A = \theta_B  = \theta_C = \theta_R = \theta_{ABC} = \theta_{ABCR} = 2.0$ and set $\theta_{AB} = 20.0$.  This amounts to slowing
 the rate of coalescence in population $AB$, which results in migration out of this population before coalescence of lineages $a$ and $b$ with fairly high probability, 
 leading to an increase in the proportion of discordant gene trees.  We assume no gene flow between species $A$ and $B$ (i.e., $m_3 = m_4 = 0$), and allow symmetric gene
 flow between populations $AB$ and $C$ at rate $m$ (i.e., we set $m_1=m_2=m$). We vary $m$ to assess performance of SVDQuartets at varying levels of
 gene flow.  Gene genealogies are simulated using the software \texttt{ms} \cite{hudson2002} with the command line:  \texttt{./ms 4 n -t 2.0 -T -I 4 1 1 1 1 -ej 0.5 1 2 -ej 1.5 2 3 -ej 5.0 3 4 -em 0.5 2 3 } $\mathbf{m/2}$ \texttt{ -em 0.5 3 2 } $\mathbf{m/2}$ \texttt{ -en 0.5 2 20.0 > treefile}, where $m=0, 0.4, 0.8, 1.2, 1.6, 2.0, 4.0,$ or $8.0$.  
 In the units used by $\texttt{ms}$, $m = 2Np$, where $p$ is the fraction of population $AB$  made up of migrants from population $C$  at each generation (and similarly for $C$ and $AB$
 since we assume symmetric migration) and $N$ is the effective population size.
 Once the gene genealogies are generated, a single site is evolved along each genealogy using the software Seq-Gen \cite{rambautgrassly1997} under either the 
 JC69 model \cite{jukescantor1969} (command: \texttt{./seq-gen -s 0.025 -q -mHKY}) or the GTR+I+G model \cite{tavare1986} (command: \texttt{./seq-gen -s 0.025 -q -mGTR -r 1.0 0.2 10.0 0.75 3.2 1.6 -f 0.15 0.35 0.15 0.35 -i 0.2}) for $$n = 10,000; 50,000; 100,000; 200,000;400,000; 600,000; 800,000;1,000,000.$$  
 At each setting for $m$ and $n$ and for each substitution model, 100 data sets are simulated, and the proportion of data sets for which the correct four-taxon tree is estimated 
 is recorded.
 
The next simulation setting we consider within the context of coalescent independent sites is designed to assess the impact of changing the $\theta$ parameters in 
the presence of moderate levels of gene flow. 
Holding $m$ fixed at 0.8, we considered $\alpha_C = 5.0, 10.0,$ and $20.0$ (recall that $\theta_C = 2/\alpha_C$).  
The corresponding command in \texttt{ms} was
\texttt{./ms 3 n -t 2.0 -T -I 3 1 1 1 -ej 0.5 1 2 -ej 1.5 2 3 -n 3 } $\mathbf{w}$ \texttt{ -em 0.5 2 3 0.4 -em 0.5 3 2 0.4 -en 0.5 2 } $\mathbf{y}$ \texttt{ > treefile}, where in this command, the parameters $\mathbf{(w,y)}$ took
values (0.1, 50.0), (0.05, 200.0), and (0.025, 800.0). For comparison, we include the case where $\alpha_C = 1.0$, which corresponds to the Simulation 1 setting for which $m=0.8$.
  All other settings were the same as in the previous simulation.  Finally, we considered a collection of settings for which there are no 
anomalous gene trees, but for which all three gene trees have somewhat similar probabilities.   The only change from the first simulation study is to set $\theta_{AB} = 4.0$ (instead of 20.0).

 \subsubsection*{Results of Coalescent Independent Sites Simulations.}
The results of the first set of simulation studies using coalescent independent sites are shown in Figure \ref{fig:cis.sim}a  for data simulated with both the JC69 model (solid lines) and the GTR+I+G model (dotted lines).  
 We first note that, in general, as the number of sites increases, the accuracy of SVDQuartets increases
as well, with accuracy at or near 100\% for 1,000,000 sites. While this is true for most levels of gene flow ($m$), 
even for the largest $m$ considered ($m=8.0$), it 
does not appear to be true for the intermediate value $m=0.4$, 
where the accuracy appears to fluctuate around $1/3$ across all sample sizes.

\begin{figure}
\includegraphics[width = 7.5in]{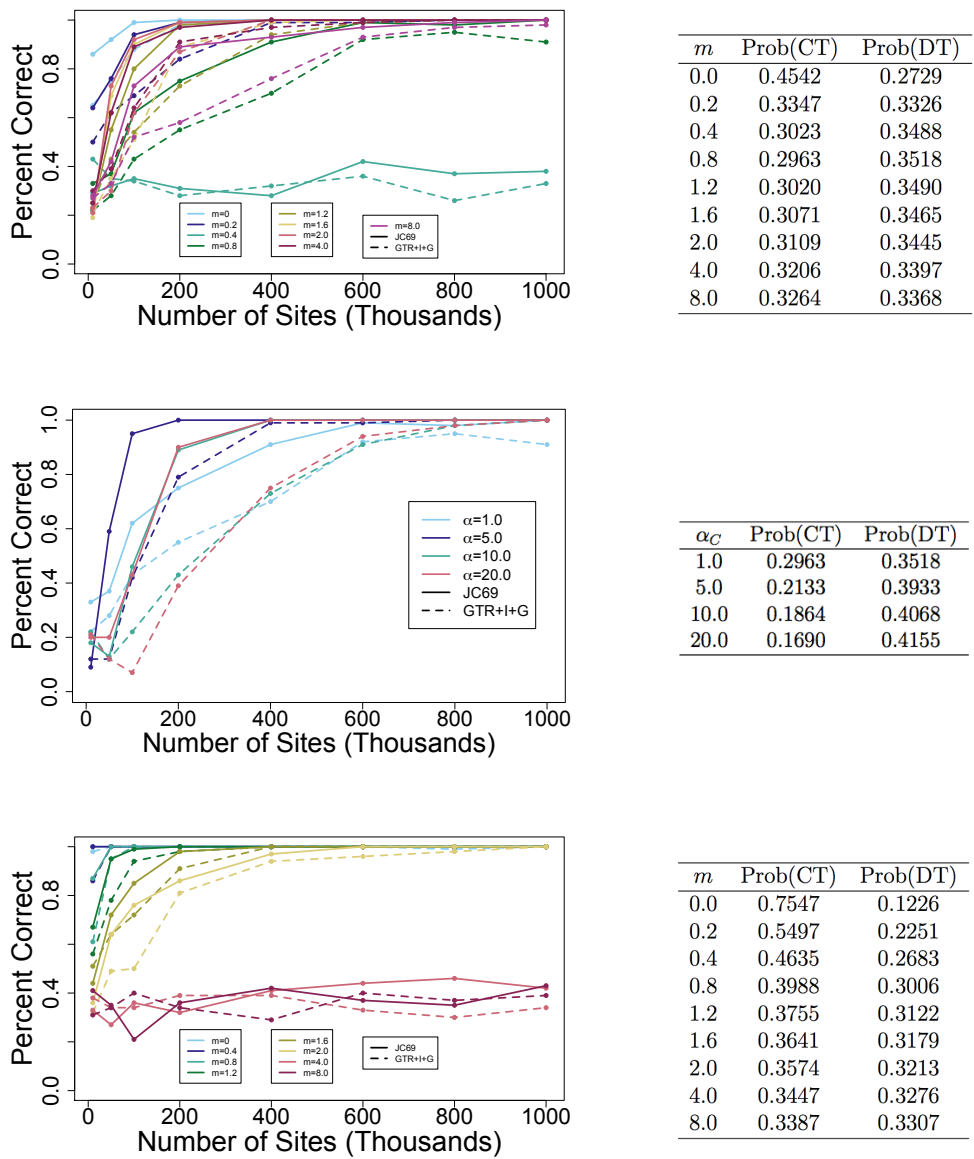}
\vspace{-2.25in}
\caption{\small Results of the simulation studies using coalescent independent sites.  In each row, the figure on the left shows the number of coalescent independent sites 
on the x-axis, and the accuracy (percent correct trees in 100 replicates) on the y-axis.  The  solid lines show results for the JC69 model, while
the  dotted lines show results for the GTR+I+G model. The table to the right gives the probability of the concordant and discordant topologies
for the values of the parameter being varied in that simulation.  Row (a): Simulation 1 varies $m$ for $\theta_{AB} = 20.0$; Row (b): Simulation 2 varies $\alpha_C$ for $m=0.8$; Row (c): Simulation 3 varies $m$ for $\theta_{AB} = 4.0$.}\label{fig:cis.sim}
\end{figure}

We investigated this unexpected result in two ways.  First, the theory above tells us that SVDQuartets is theoretically valid 
for all values of $m$, and so, with enough
data, the method should be able to accurately infer the species tree.  Thus, we repeated the simulation with $n=10,000,000$ bp, and found that 
SVDQuartets recovered the correct tree 79\% of the time 
for the JC69 model and 61\% of the time for the GTR+I+G model.  These results suggest that there is some feature of this model that makes inferring the species tree under these conditions more difficult,  in the sense that more data are required for accurate inference. Our second approach was therefore
to try to understand why this might be the case.  We considered the predicted frequencies of the gene tree histories 
(computed using the COALGF software \cite{tiankubatko2016}) in the context of the theory in the previous section that shows
 the validity of SVDQuartets. In particular,
we note that  the signal for inferring the species tree comes from the expected equality of site patterns $p_{i_1i_2i_3i_4}$ and $p_{i_2i_1i_3i_4}$, but 
that many different histories provide observations of these site patterns.  For this choice of parameters, the particular histories that carry most of the information about the species tree
occur in very low frequency, and thus large sample sizes are required for the signal to become definitive.   While 
this choice of parameters is 
especially difficult for SVDQuartets in that more data are required for good performance, the range of these ``bad'' parameter choices appears to be narrow, since the setting $m=0.2$ (navy lines in Figure  \ref{fig:cis.sim}a) and $m=0.8$ (dark green lines in Figure \ref{fig:cis.sim}a) both show good performance.

Our next observation from the results displayed in Figure  \ref{fig:cis.sim}a is that, in general, when data are simulated under the GTR+I+G model, 
more data are needed 
to achieve a given level of accuracy.  This is reasonable, since more complicated models may require more data for accurate inference. Finally, we note that as $m$ increases, 
the method again requires more data for accurate inference (with the exception noted above).  While the values of $n$ (the number of coalescent 
independent sites) considered here are large, many of those sites will be constant or non-informative.  To emphasize this, Figure \ref{fig:site.pattern.hists}  shows histograms of the number of 
variable sites (Figures \ref{fig:site.pattern.hists}a and c) and number of parsimony informative sites (Figures  \ref{fig:site.pattern.hists}b and d) for the cases where $m=0$ and where $m=2.0$.  When $m=0$, the mean proportion of variable
sites across 100 replicates under the JC69 model was 30.13\%, and the mean proportion of parsimony-informative sites was 3.97\%.  When $m=2.0$, the mean proportion of variable
sites was 26.19\%, and the mean proportion of parsimony-informative sites was 4.23\%. Thus, the number of SNPs required (using the standard definition of SNPs) would be much lower than the values of $n$ considered here.

\begin{figure}
\begin{center}
\includegraphics[width = 5.5in]{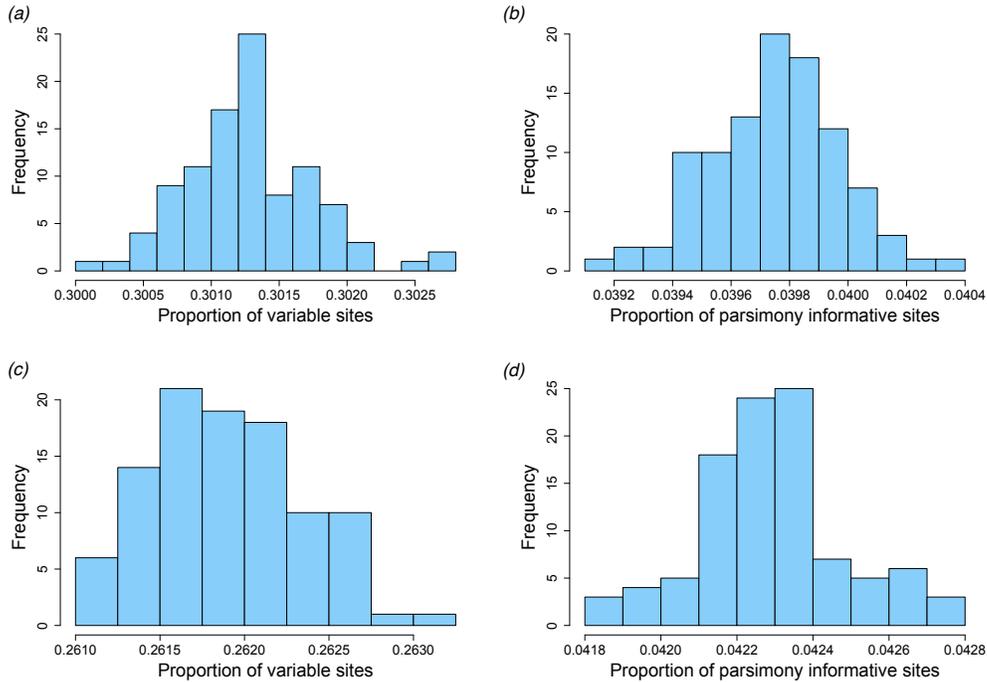}
\caption{\small Proportion of variable sites (a and c) and parsimony-informative sites (b and d) for the coalescent independent sites simulations under the JC69 model 
when $m=0$ (a and b) and when $m=2.0$ (c and d).}\label{fig:site.pattern.hists}
\end{center}
\end{figure}

The results of the second simulation study, in which population sizes are varied while $m$ is fixed at 0.8, are shown in Figure \ref{fig:cis.sim}b.  The general observations 
in this case are the same as above.  In particular, as the sample size $n$ increases, the accuracy of the method increases as well, with accuracy at or near 100\% for the largest
sample sizes considered. Also, the accuracy for data simulated under the GTR+I+G model is lower at a fixed $n$ than for those simulated under JC69, a finding not unexpected
given that GTR+I+G is a more complicated model than JC69.  
Finally, we point out that SVDQuartets is successful even in the case of extreme incongruence.
 For example, for the largest $\alpha_C$  
considered ($\alpha_C = 20.0$), 
the probability of the gene tree topology that matches the species tree is only 16.9\%, while each of the other two gene tree topologies have probability  41.55\%.  Thus, though
fewer than 
20\% of the data points are obtained from the gene tree matching the species tree, SVDQuartets is able to recognize the patterns in the data once a sufficient 
amount of data are observed
and correctly infer the species tree with high accuracy.

In the third simulation study using coalescent independent sites, there are no anomalous gene trees; instead, the probabilities of each of the three gene trees are relatively 
similar, with the probability of the gene tree that matches the species tree a little bit larger.   The results of this simulation are shown in Figure \ref{fig:cis.sim}c, and
indicate that the same basic patterns of the previous simulations are maintained.  Specifically, accuracy increases with $n$ and decreases as we move from the simpler model
(JC69) to the more complex model (GTR+I+G).  We also note that accuracy decreases as $m$ increases, indicating that the addition of gene flow to the model makes it more
difficult to infer the correct species tree.  However, for all but the largest values of gene flow ($m=4.0$ and $m=8.0$) the accuracy is over 90\% with 400,000 or more sites, while
for $m=4.0$ and $m=8.0$, the accuracy is only around 30\% at all levels of $n$ considered.  This level of gene flow perhaps represents the case where the extent of genetic exchange between populations is large enough that a species tree cannot be determined when all three trees are observed with approximately equal frequency.  Interestingly,
however, in the first simulation study, these high levels of gene flow did not hinder accurate estimation. We conjecture that the signal in the case of anomalous gene trees
in the first simulation study makes it easier to identify the correct species tree, whereas in this case, all three trees appear in near equal frequency.  In addition, the role
of the mutation process and distribution of coalescent times likely play a role in determining the difficulty of the inference problem, and these are more difficult to examine.
However, the theory above does indicate that with enough data, SVDQuartets should be able to accurately infer the species tree, even when the level of gene flow is high.  To
examine the issue in more detail, we repeated the simulation for $m=8.0$ with $n=10,000,000$ sites for the JC69 model, and found that the proportion of correctly estimated trees had risen to 54\%, compared to 39\% when $n = 1,000,000$, indicating that with sufficient data, the method
can correctly recover the species tree.

\subsubsection*{Multilocus Data.} We repeat all three of the simulation studies described above in the setting in which multilocus data, rather than coalescent independent sites, have been collected.  The first step of each simulation is the same as described above: the \texttt{ms} software is used to simulate a set of gene trees from the model
species tree in Figure \ref{fig:sim.model}.  The number of gene trees simulated, denoted $G$, will correspond here to the number of genes, and we consider $G$ ranging from 200 to 3,000 
in increments of 200.  Then, for each gene tree, an alignment of 300bp is simulated using Seq-Gen, with the same settings as above. As before, 100 replicates at each
simulation setting are generated. We analyze each simulated multilocus data set with SVDQuartets, as well as with two other software packages designed to estimate
species trees from multilocus data under the coalescent model, ASTRAL \cite{mirarabetal2014} and MP-EST \cite{liuetal2010}.   As mentioned above, both ASTRAL and MP-EST are summary methods that work by
first estimating gene trees for each gene, and then looking for subtrees that appear most frequently in the set of estimated gene trees. 
ASTRAL uses (unrooted) quartet frequencies across 
gene trees to find the species tree, while MP-EST considers rooted triples.  In either case, we expect that ASTRAL and MP-EST will have 
difficulty in estimating the species tree in the
presence of anomalous gene trees, because incorrect relationships will appear at highest frequency across the genes in this case.    In our simulations, we use PAUP* to estimate the gene
trees using maximum likelihood with the GTR+I+G model with all parameters estimated.  In the case of ties in the likelihood score, we randomly keep
one of the tied trees using the PAUP* option \texttt{keepone=random}. The gene trees estimated by PAUP* are used as input into ASTRAL and MP-EST.

\subsubsection*{Results of the Multilocus Simulations.}
The results of the first multilocus simulation are shown in Figure \ref{fig:multi.sim}a and \ref{fig:multi.sim}b for the JC69 model and GTR+I+G model, respectively, for varying
 $m$. First, we note 
that the accuracy of SVDQuartets (solid lines) is similar to that of Simulation 1 for the coalescent independent sites data, except that the number of sites required to achieve 
the same accuracy in the multilocus setting is larger. This is because, for example, 600 sites sampled from 2 genes (300bp
per gene, as in our multilocus simulation study) give less
information about the species tree than 600 sites sampled from 600 gene trees (as in our coalescent independent sites simulations).  The 600 sites observed from 600
distinct gene trees give independent genealogical information about the species tree, though indirectly, while the 300 sites for each of the two genes can give a reasonable indication of the individual gene trees, but still provide only two observed gene genealogies. 

Next, we note that the number of genes needed to achieve a given level of accuracy is larger for the GTR+I+G model than for the JC69 model, similar to what was observed for the 
coalescent independent sites simulations.   Comparing the performance of ASTRAL and MP-EST with SVDQuartets, we see that all three methods perform well when there is no
gene flow or when the level of gene flow is low ($m=0.2$), with ASTRAL and MP-EST outperforming SVDQuartets when the number of genes is not too large. However, as the 
amount of gene flow increases, both ASTRAL and MP-EST have accuracy that decreases to 0\% as the number of genes increases, as expected. It is interesting to note that the
unusual portion of the parameter space that was observed in the coalescent independent sites simulations (i.e., $m=0.4$) exists here as well -- and all three methods struggle 
to obtain accurate inference at smaller sample sizes under this setting.  We maintain the same explanation as above for the
performance of SVDQuartets, and note that we expect the accuracy of SVDQuartets to increase as the data size increases. To test this, we 
repeated our simulation with 50,000 genes of length 300bp under the JC69 model, and observed that 60\% of the species tree were accurately
estimated, supporting our explanation that larger sample sizes are need for more challenging problems.
However,  it is somewhat  puzzling that ASTRAL and MP-EST do as well as they do in this case, since the gene tree that matches the species tree should 
be estimated less frequently than the other two gene trees.  We conjecture that this results from an interaction between the mutation process and the process of generating the gene
genealogies, and that this effect would disappear if longer (i.e., $> 300$bp), more informative genes were generated.   We note that our observation is similar to other observations
 in which species trees in the anomaly zone do not necessarily lead to inconsistent inference, while those outside the anomaly zone may \cite{kubatkodegnan2007,huangknowles2009}. In other words, the anomaly zone is not the only predictor of performance of species tree inference methods; rather, the mutation 
process also plays an important role.

\begin{figure}[!h]
    \centering
    \begin{minipage}{.5\textwidth}
        \centering
        {\includegraphics[width = 3.0in]{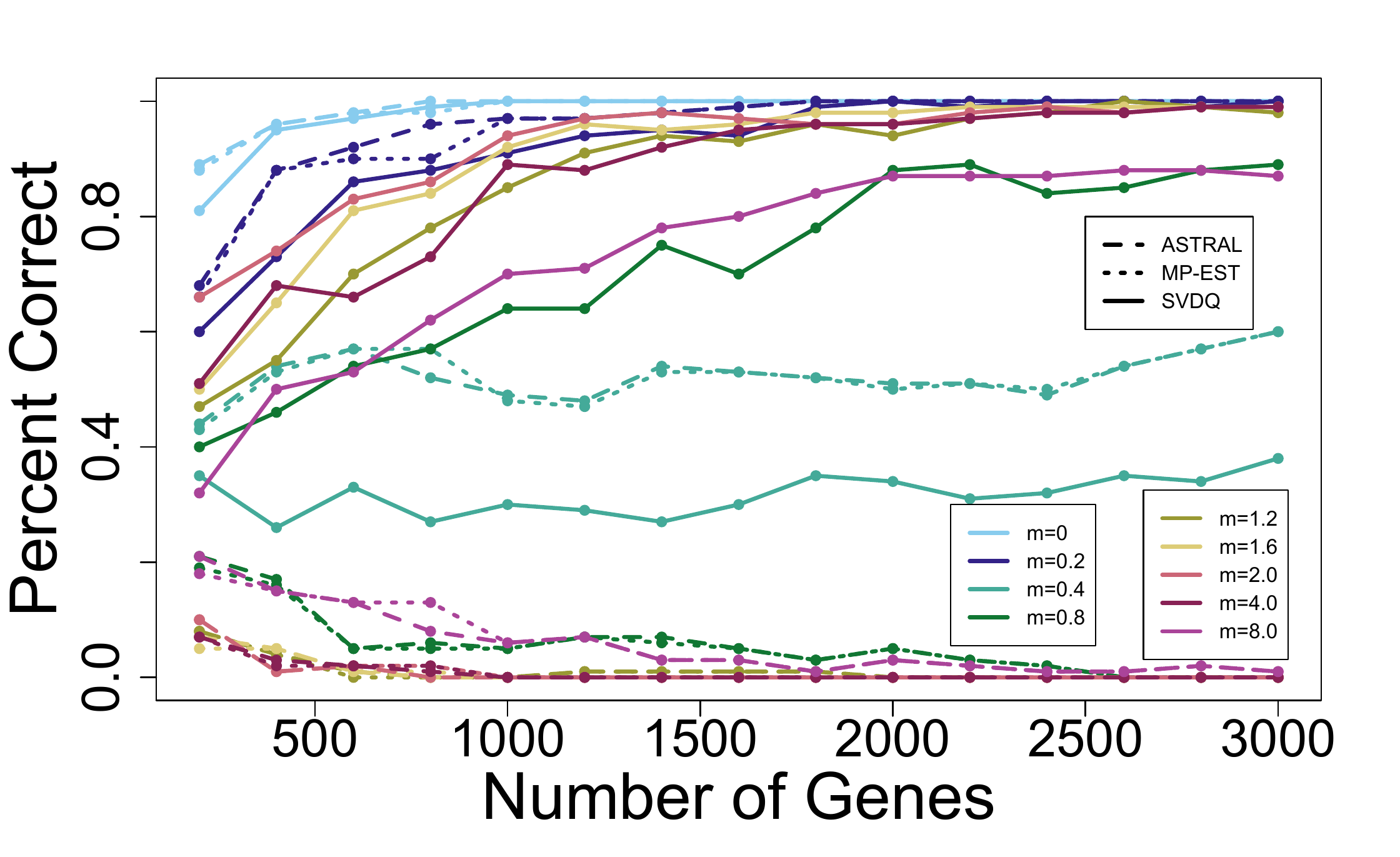}}\\
        (a)
        \label{fig:prob1_6_2}
    \end{minipage}%
    \begin{minipage}{0.5\textwidth}
        \centering
       {\includegraphics[width=3.0in]{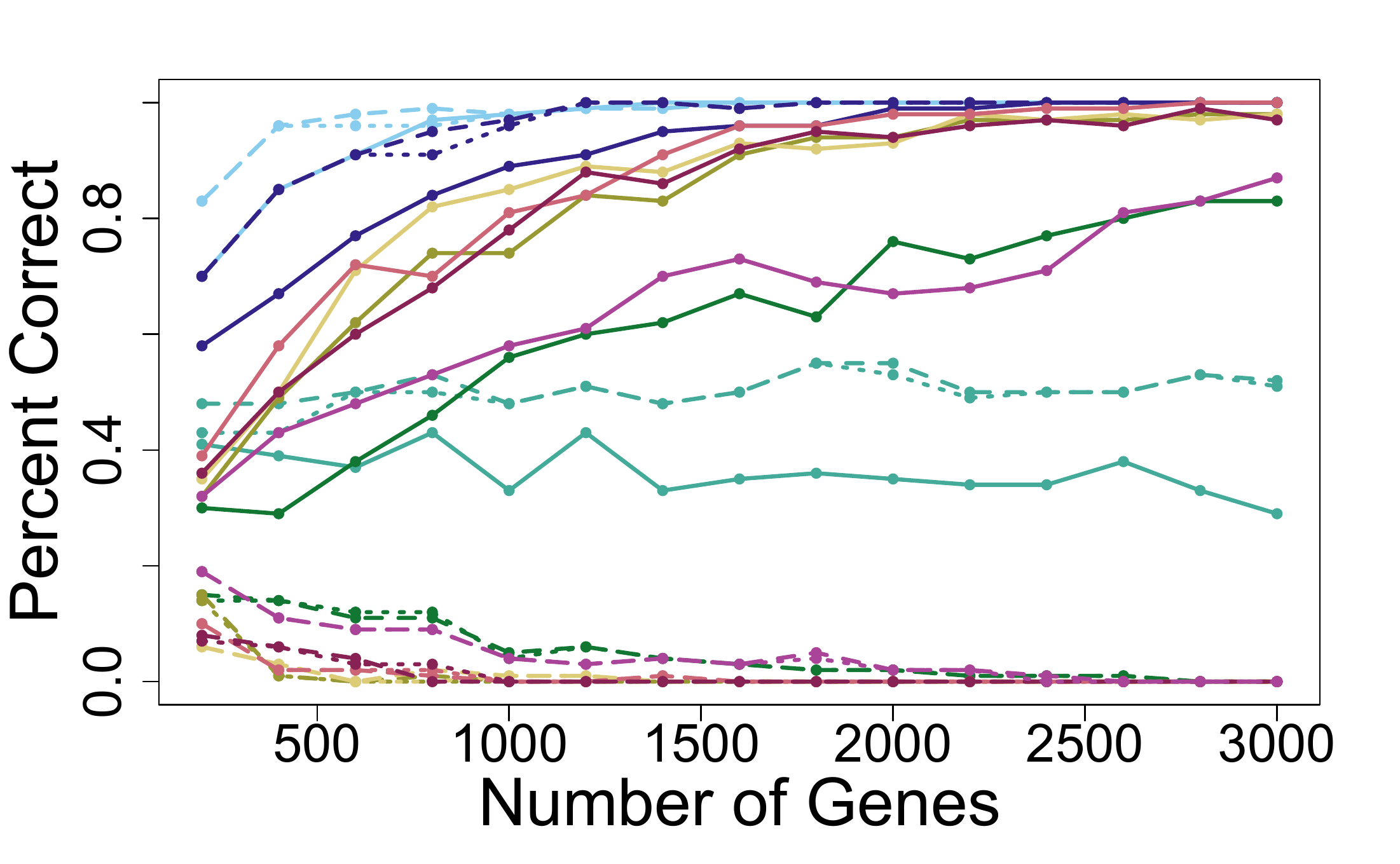}}\\
       (b)
              \label{fig:prob1_6_1}
    \end{minipage}

    \centering
    \begin{minipage}{.5\textwidth}
        \centering
        {\includegraphics[width = 3.0in]{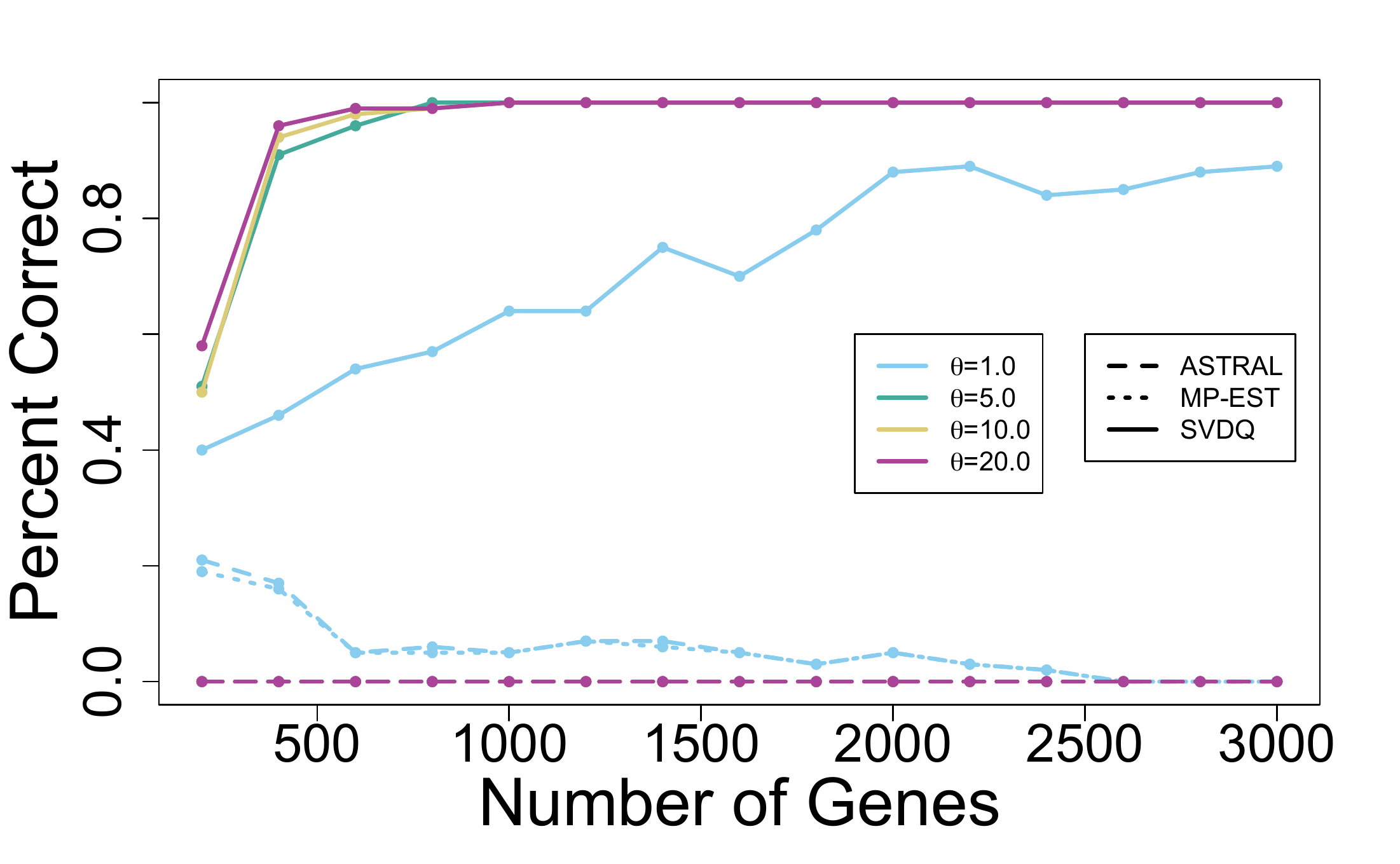}}\\
        (c)
               \label{fig:prob1_6_2}
    \end{minipage}%
    \begin{minipage}{0.5\textwidth}
        \centering
       {\includegraphics[width=3.0in]{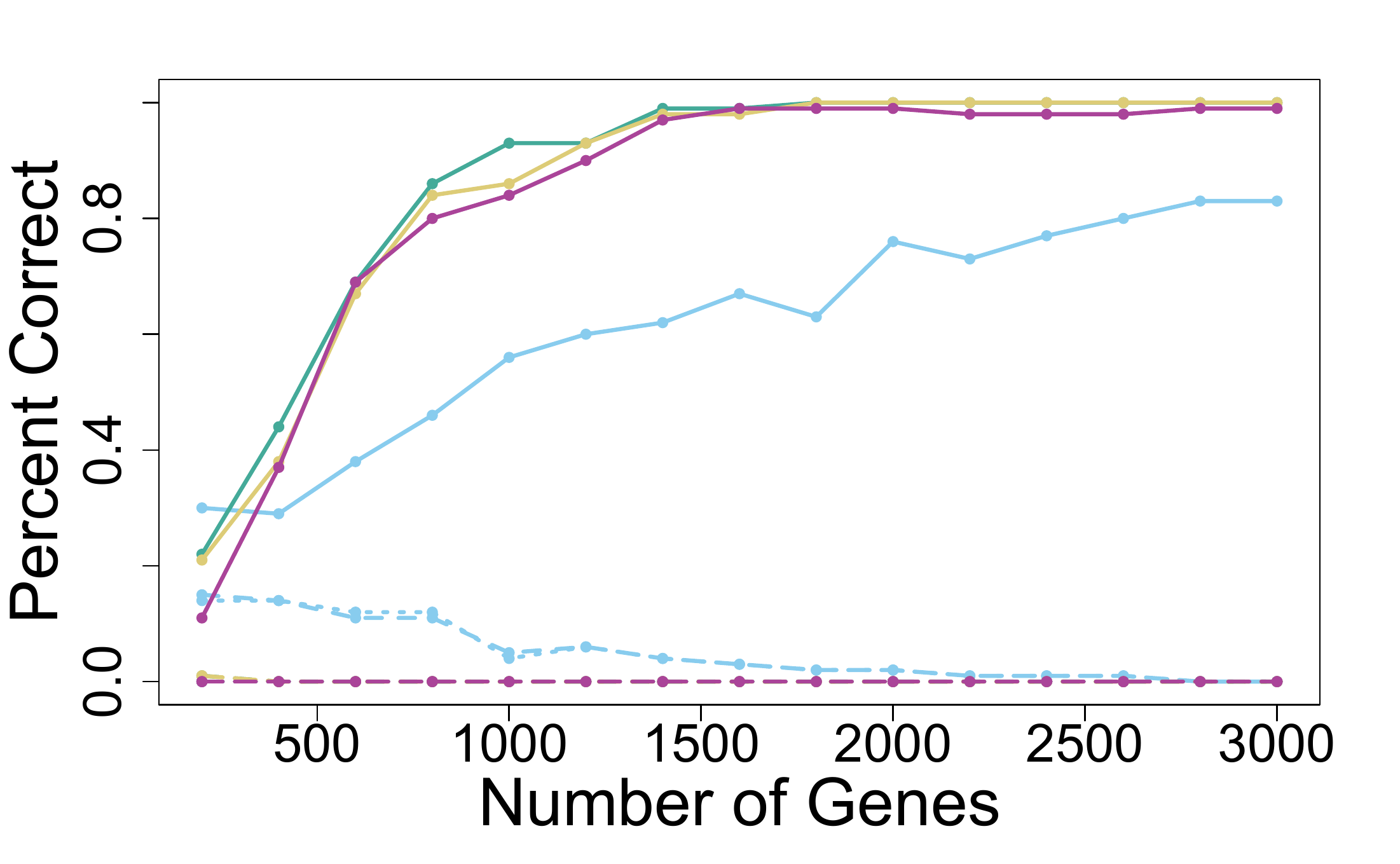}}\\
       (d)
        \label{fig:prob1_6_1}
    \end{minipage}

    \centering
    \begin{minipage}{.5\textwidth}
        \centering
        {\includegraphics[width = 3.0in]{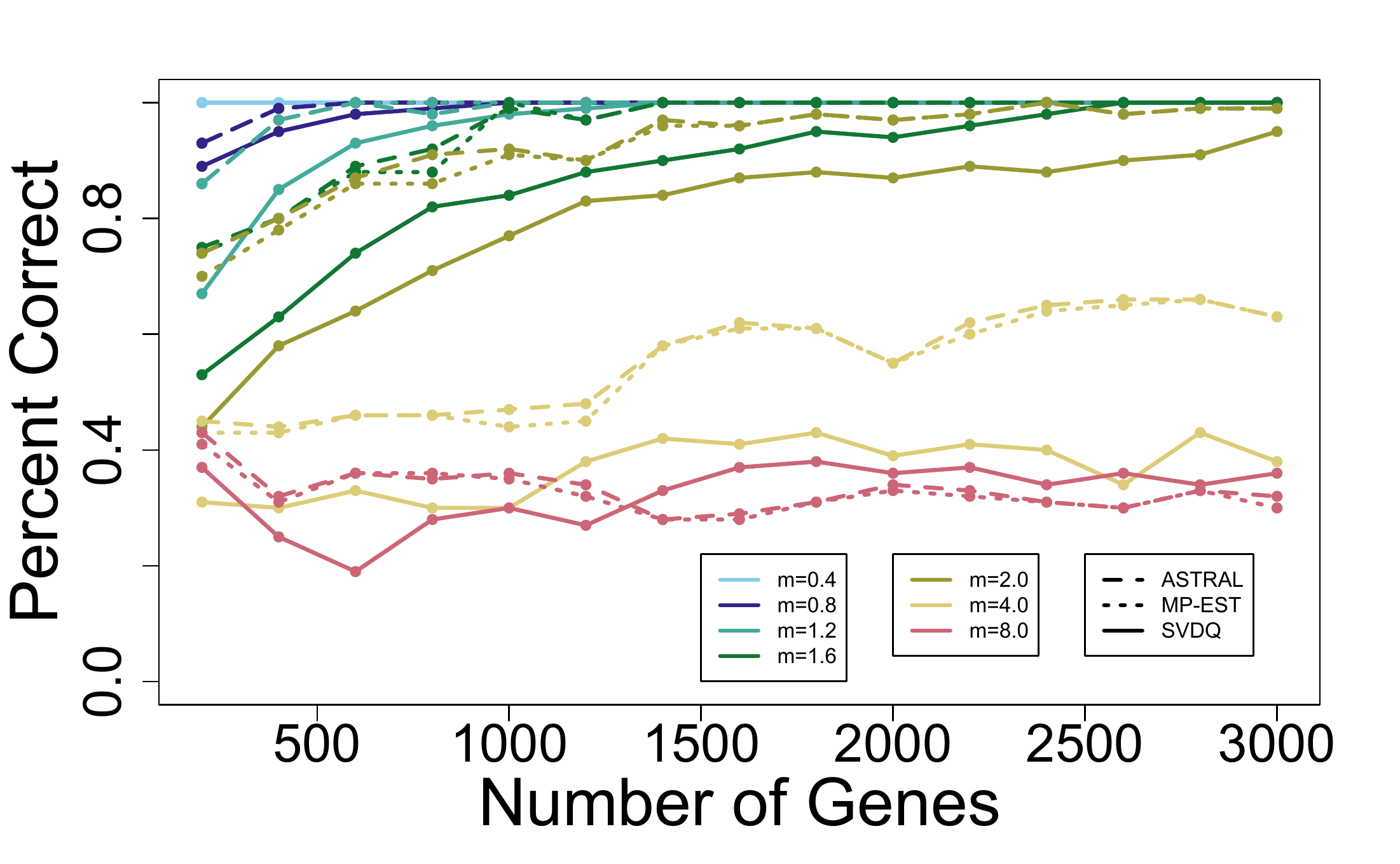}}\\
        (e)
        \label{fig:prob1_6_2}
    \end{minipage}%
    \begin{minipage}{0.5\textwidth}
        \centering
       {\includegraphics[width=3.0in]{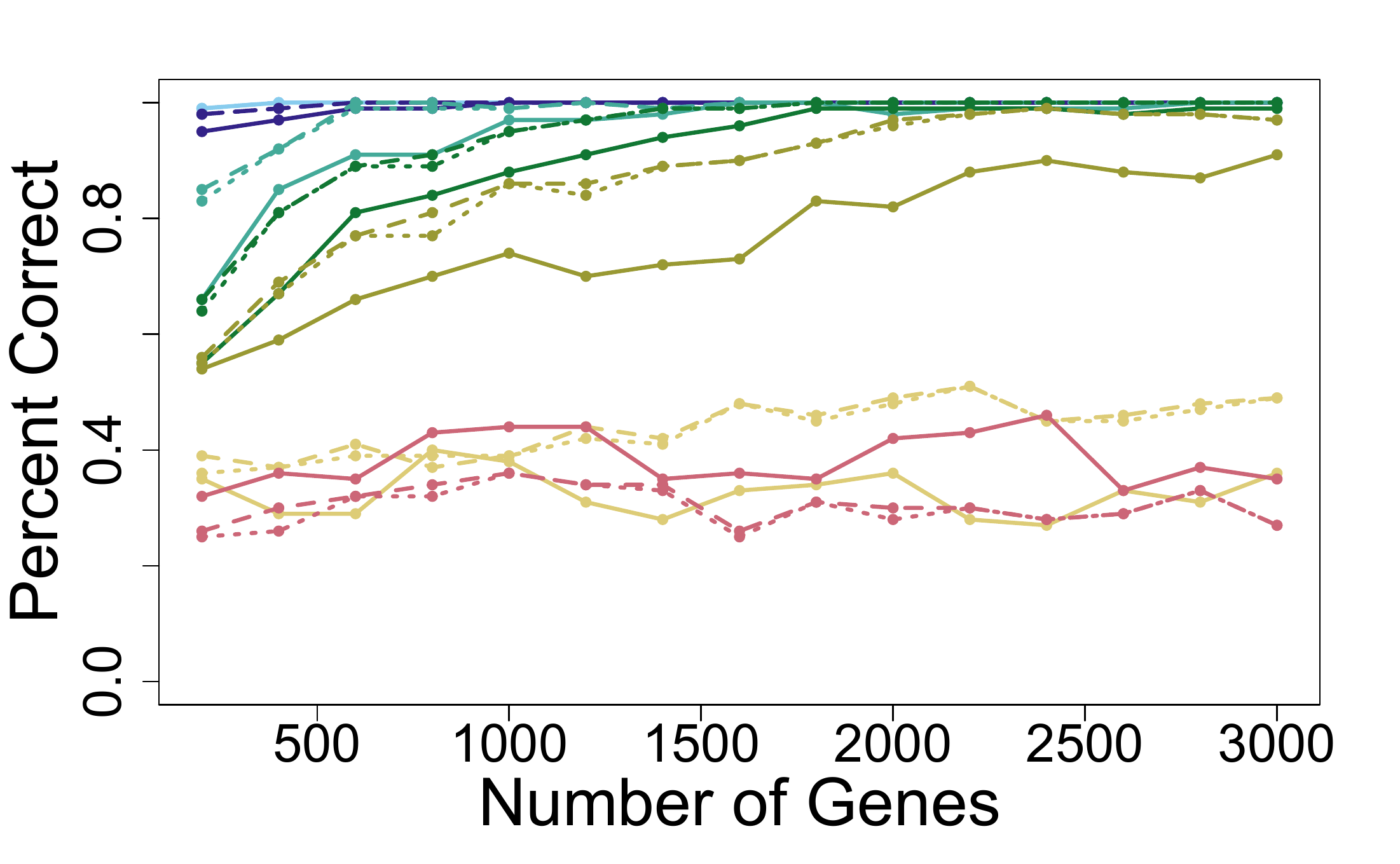}}\\
       (f)
               \label{fig:prob1_6_1}
    \end{minipage}
    \caption{Results of the multilocus simulation study.  The rows give the results of simulations 1, 2, and 3, respectively, while
    the left column gives results for the JC69 model and the right column gives results for the GTR+I+G model. In each 
    graph, the x-axis gives the number of genes, each 300bp in length, and the y-axis shows the percent of trees correctly 
    estimated in 100 replicates.  Parameter settings for simulations 1, 2, and 3 are as in Figure \ref{fig:cis.sim}.}\label{fig:multi.sim}
\end{figure}

The results of the second multilocus simulation study are shown in Figure \ref{fig:multi.sim}c (JC69 model) and  \ref{fig:multi.sim}d (GTR+I+G model) when $m$ is fixed at 0.8 and $\alpha_C$ varies. Since there are anomalous gene trees in this case, we expect ASTRAL and MP-EST to perform poorly, and this is in fact 
what is observed. 
Their accuracy goes to 0\% quickly for all choices of $\alpha_C$. SVDQuartets also behaves as expected, with accuracy increasing to 100\% as sample size, in this case, the number of 
genes, increases, for all choices of $\alpha_C$.

The  results of third multilocus simulation study are shown in Figures \ref{fig:multi.sim}e (JC69 model) and \ref{fig:multi.sim}f (GTR+I+G model). Recall 
that in this case, there are no anomalous gene trees, but that the three gene tree topologies occur in relatively equal frequency.  For both the JC69 and the GTR+I+G 
models, all three methods perform well when the level of gene flow is not too high ($m\leq2.0$), but their performance deteriorates when $m=4.0$ or $m=8.0$.  These results are 
consistent with what was observed for the coalescent independent sites data.

\section*{Discussion}
\label{sec: Discussion}
Our goals in this paper are two-fold: first, to understand how gene flow affects 
the distribution of gene tree topologies under the coalescent; and second, 
to assess how gene flow impacts the accuracy of species tree estimation methods. 
Towards the first goal, we have shown that in a simple 3-taxon species tree, 
the combination of gene flow and variation in effective population sizes across the tree
can result in anomalous gene trees.  
Under a simplified model in which gene flow occurs only between ancestral sister populations, 
we derived the probability of obtaining a gene 
tree concordant with the species tree, and characterized the {\itshape gene flow anomaly zone},
the regions of parameter space in which the two discordant gene trees each 
have higher probability than that of the concordant gene tree. 
Our results show that when asymmetric migration rates are allowed, the probability of the concordant topology can be made arbitrarily small, while the probabilities of each of the two discordant topologies approach 50\%.  
However, when symmetric gene flow is assumed, the probability of the concordant topology is bounded below by $1/9$, with values
between $1/9$ and $1/3$ resulting in anomalous gene trees.   In addition, in the case of symmetric gene flow, we found that the rate of coalescence in the recipient population must be more than 3 times as large as that of the donor population (looking backward in time) for anomalous gene trees to exist under this simplified model.

Figure \ref{fig: SurfacePlot} highlights that, though the model is simple, the gene flow anomaly zone can be fairly complicated due to the interaction between the migration rate parameters and
the population sizes.  The lengths of the speciation intervals also play a role, as these interact with the coalescent rate to determine the probability with which coalescent events occur within specific populations on the species tree. 
While our results apply directly to the reduced IM model from Figure \ref{fig: reducedIMmodel}, we can extend these
to conclude that there are still anomalous gene trees in the full 3-taxon IM model in Figure \ref{fig: IMmodel}.
The main differences between these models is that the 3-taxon IM model includes an interval of gene flow between populations $A$ and $B$ and an interval of no gene flow between populations $AB$ and $C$. 
During each of these intervals, the only coalescent events possible are those leading to a gene tree that is concordant 
with the species tree. Thus, with all other parameters fixed, the probabilities of the concordant topology will increase 
as the length of either of these intervals increases (with a corresponding decrease in the proportion of discordant topologies). Still, as the length of these intervals goes to zero, the model approaches the reduced IM model and 
there will be choices of parameters producing anomalous gene trees. Of course, if the interval of gene flow between populations $A$ and $B$ is decreased and the interval of no gene flow between populations $AB$ and $C$ is increased, the model approaches the standard model with no gene flow, for which it is known that there are no anomalous gene trees. 

Towards our second goal, we argued that SVDQuartets is a theoretically valid method
in the presence of gene flow and supported this result with simulation studies.
 For example, even in the case in which the probability of the concordant tree was only $\sim 17\%$ (see Figure \ref{fig:cis.sim}b), 
SVDQuartets correctly inferred the species tree with over 95\% accuracy for 400,000 coalescent independent sites under the JC69 model, and 
for 800,000 coalescent independent sites for the GTR+I+G model.  While SVDQuartets shows similar accuracy when 
multilocus data are simulated, ASTRAL and MP-EST often perform poorly in the gene flow anomaly zone, with accuracy decreasing to 0\% as the number of genes increases. This is to
be expected in a sense, because both ASTRAL and MP-EST are based on models that assume immediate cessation of gene flow following speciation and constant effective population size.   Our simulations thus indicate that these methods can be extremely sensitive to violations 
of these assumptions.

Certain model conditions were difficult for all three of the methods. In particular,
very high rates of gene flow resulted in poor accuracy for all methods with 3,000 genes, even for parameter choices outside of the gene flow anomaly zone. This observation is consistent with
studies of the performance of species tree estimation methods in the anomaly zone \cite{kubatkodegnan2007,huangknowles2009}, for which it has been observed that there is also
an important effect of the mutation process on accuracy.  
We also note that in the presence of very high levels of gene flow, it may even be argued that species should not actually 
be considered distinct.
But even in these situations, it is reassuring that with sufficient data, SVDQuartets can overcome the conflicting signal in the data and correctly
infer the true relationships.

The model in Figure \ref{fig: IMmodel} is of particular interest, because it represents the biologically plausible scenario in which speciation occurs with subsequent gene flow between populations 
for some period of time, which may be more realistic than speciation with instantaneous cessation of gene flow.  This was our rationale for focusing on a model that only allowed 
gene flow between sister populations for a specific time following speciation.  The fact that SVDQuartets remains theoretically valid under this model means that it can be confidently applied to
empirical problems for which gene flow is thought to have occurred between any pair of sister taxa. Although SVDQuartets sometimes requires a large amount of data when the underlying
substitution model is complex, it appears to be statistically consistent even in these complicated scenarios. In practice, bootstrap support values would provide a reasonable measure of 
uncertainty in the species tree estimate, 
and would be expected to be low when the amount of data was insufficient.  
In contrast, methods like ASTRAL and MP-EST could be statistically 
inconsistent, and bootstrap support values would likely not reflect this, because they assume an incorrect model (e.g., they assume speciation with immediate cessation of gene flow and constant effective population size).  Thus, ASTRAL and MP-EST could be 
said to be positively misleading, in the sense of \cite{kubatkodegnan2007}, while SVDQuartets provides 
inference with an appropriate quantification of uncertainty.

Though we have focused our discussion largely on the effect of gene flow, we note that the model in Figure \ref{fig: IMmodel} generalizes previous species tree estimation frameworks
by allowing variation in the effective population sizes, which determine the rate of coalescence across the species tree, even in the absence of gene flow.  Further, we note that the results presented here
concerning accuracy of the SVDQuartets method also apply to the 
case in which the species tree does not satisfy the molecular clock, following the arguments provided in \cite{longkubatko2017}.  Thus SVDQuartets holds under very general conditions, as 
it allows variation in effective population sizes across the species tree, violation of the molecular clock, and gene flow between sister taxa following speciation.

One natural extension of the model in Figure \ref{fig: IMmodel}  is to allow gene flow between all pairs of taxa, including non-sisters. We have not yet considered this case, but speculate that
certain choices of gene flow rates and effective population sizes might be problematic even for SVDQuartets, though this warrants further exploration.  Though we consider here only the case of 
four-taxon species trees in our simulation studies, we note that for SVDQuartets this is sufficient, since the overall species-level phylogenetic estimate is obtained through assembly of the 
inferred relationships among quartets. For ASTRAL and MP-EST, however, we expect that difficulties in the four-taxon case will translate into reduced accuracy on larger trees, as well.  

Models for estimation of species-level phylogenies based on data collected across the genome must 
necessarily try to incorporate realistic evolutionary mechanisms.  
By examining a model of speciation with gene flow, we have highlighted the complexity involved in using 
genealogical data to infer species-level relationships, as well as the impact that this complexity has on some of the commonly-used methods of species tree inference. Our finding that the 
method on which SVDQuartets is based holds for this more general model reinforces that this method is a valuable tool for estimating species trees from genome-scale data.  Future work
should continue to focus on improving the biological realism on 
which methods for species tree estimation are based.

\bibliography{bibfile}

\begin{thebibliography}{10}

\bibitem{andersenetal2014}
L.~N. Andersen, T.~Mailund, and A.~Hobolth.
\newblock {Efficient computation in the IM model}.
\newblock {\em Journal of Mathematical Biology}, 68(6):1423--1451, 2014.

\bibitem{chifmankubatko2014}
J.~Chifman and L.~Kubatko.
\newblock Quartet inference from {SNP} data under the coalescent model.
\newblock {\em Bioinformatics}, 30(23):3317--3324, 2014.

\bibitem{chifmankubatko2015}
J.~Chifman and L.~Kubatko.
\newblock Identifiability of the unrooted species tree topology under the
  coalescent model with time-reversible substitution processes, site-specific
  rate variation, and invariable sites.
\newblock {\em Journal of Theoretical Biology}, 374:35--47, 2015.

\bibitem{chungane2011}
Y.~Chung and C.~An{\'e}.
\newblock {Comparing two Bayesian methods for gene tree / species tree
  reconstruction: A simulation with incomplete lineage sorting and horizontal
  gene transfer}.
\newblock {\em Systematic Biology}, 60(3):261--275, 2011.

\bibitem{degnansalter2005}
J.~Degnan and L.~Salter.
\newblock Gene tree distributions under the coalescent process.
\newblock {\em Evolution}, 59:24--37, 2005.

\bibitem{eckertcarstens2008}
A.~J. Eckert and B.~C. Carstens.
\newblock Does gene flow destroy phylogenetic signal? the performance of three
  methods for estimating species phylogenies in the presence of gene flow.
\newblock {\em Molecular Phylogenetics and Evolution}, 49(3):832--842, 2008.

\bibitem{heleddrummond2010}
J.~Heled and A.~J. Drummond.
\newblock Bayesian inference of species trees from multilocus data.
\newblock {\em Molecular Biology and Evolution}, 27(3):570--580, 2010.

\bibitem{hey2010}
J.~Hey.
\newblock Isolation with migration models for more than two populations.
\newblock {\em Molecular Biology and Evolution}, 27(4):905--920, 2010.

\bibitem{heyneilsen2004}
J.~Hey and R.~Nielsen.
\newblock {Multilocus methods for estimating population sizes, migration rates
  and divergence time, with applications to the divergence of Drosophila
  pseudoobscura and D. persimilis}.
\newblock {\em Genetics}, 167:747?760, 2004.

\bibitem{hobolthetal2011}
A.~Hobolth, L.~N. Andersen, and T.~Mailund.
\newblock On computing the coalescence time density in an
  isolation-with-migration model with few samples.
\newblock {\em Genetics}, 187(4):1241--1243, 2011.

\bibitem{huangknowles2009}
H.~Huang and L.~L. Knowles.
\newblock {What is the danger of the anomaly zone for empirical phylogenetics?}
\newblock {\em Systematic Biology}, 58(5):527--536, 2009.

\bibitem{hudson2002}
R.~R. Hudson.
\newblock {Generating samples under a Wright--Fisher neutral model of genetic
  variation}.
\newblock {\em Bioinformatics}, 18(2):337--338, 2002.

\bibitem{jukescantor1969}
T.~Jukes and C.~R. Cantor.
\newblock Evolution of protein molecules.
\newblock In H.~N. Munro, editor, {\em Mammalian protein metabolism}, pages
  21--123. Academic Press, New York, 1969.

\bibitem{karlintaylor1975}
S.~Karlin and H.~M. Taylor.
\newblock {\em {A First Course in Stochastic Processes}}.
\newblock Elsevier, second edition, 1975.

\bibitem{kubatkodegnan2007}
L.~Kubatko and J.~Degnan.
\newblock Inconsistency of phylogenetic estimates from concatenated data under
  coalescence.
\newblock {\em Systematic Biology}, 56:17--24, 2007.

\bibitem{leacheetal2014}
A.~D. Leach{\'e}, R.~B. Harris, B.~B.~Rannala, and Z.~Yang.
\newblock {The influence of gene flow on species tree estimation: A simulation
  study}.
\newblock {\em Systematic Biology}, 63:17--30, 2014.

\bibitem{liuetal2010}
L.~Liu, L.~Yu, and S.~V. Edwards.
\newblock A maximum pseudo-likelihood approach for estimating species trees
  under the coalescent model.
\newblock {\em BMC Evolutionary Biology}, 10(1):302, 2010.

\bibitem{liuetal2009}
L.~Liu, L.~Yu, L.~Kubatko, D.~K. Pearl, and S.~V. Edwards.
\newblock Coalescent methods for estimating multilocus phylogenetic trees.
\newblock {\em Molecular Phylogenetics and Evolution}, 53:320--328, 2009.

\bibitem{longkubatko2017}
C.~L. Long and L.~S. Kubatko.
\newblock {Identifiability and reconstructibility of species phylogenies under
  a modified coalescent}.
\newblock {\em Bulletin of Mathematical Biology, in revision}, 2017.

\bibitem{mirarabetal2014}
S.~Mirarab, R.~Reaz, M.~S. Bayzid, T.~Zimmermann, M.~S. Swenson, and T.~Warnow.
\newblock {ASTRAL: Genome-scale coalescent-based species tree estimation}.
\newblock {\em Bioinformatics}, 30(17):i541--i548, 2014.

\bibitem{rambautgrassly1997}
A.~Rambaut and N.~Grassly.
\newblock {SeqGen: An application for the Monte Carlo simulation of DNA
  sequence evolution along phylogenetic trees}.
\newblock {\em Computer Applications in Biosciences}, 13:235--238, 1997.

\bibitem{solislemusetal2016}
C.~Sol{\'i}s-Lemus, M.~Yang, and C.~An{\'e}.
\newblock {Inconsistency of species-tree methods under gene flow}.
\newblock {\em Systematic Biology}, 65(5):843--851, 2016.

\bibitem{tavare1986}
S.~Tavare.
\newblock Some probabilistic and statistical problems in the analysis of {DNA}
  sequences.
\newblock {\em Lectures on Mathematics in the Life Sciences (American
  Mathematical Society)}, 17:57--86, 1986.

\bibitem{tiankubatko2016}
Y.~Tian and L.~Kubatko.
\newblock Distribution of gene tree histories under the coalescent model with
  gene flow.
\newblock {\em Molecular Phylogenetics and Evolution}, 105:177--192, 2016.

\bibitem{wanghey2010}
Y.~Wang and J.~Hey.
\newblock Estimating divergence parameters with small samples from a large
  number of loci.
\newblock {\em Genetics}, 184(2):363--379, 2010.

\bibitem{zhuyang2012}
T.~Zhu and Z.~Yang.
\newblock Maximum likelihood implementation of an isolation-with-migration
  model with three species for testing speciation with gene flow.
\newblock {\em Molecular Biology and Evolution}, 29(10):3131--3142, 2012.

\end{thebibliography}
\bibliographystyle{abbrv}

\end{document}